\documentclass[submission,copyright,creativecommons]{eptcs}
\providecommand{\event}{SynCoP 2014} % Name of the event you are submitting to
\usepackage{breakurl}             % Not needed if you use pdflatex only.
\usepackage{ifthen}
\usepackage{tikz}
\usetikzlibrary{automata,arrows,intersections,positioning,calc,shapes}
\tikzstyle{state}=[circle,draw=black,inner sep=0pt,minimum size=10pt]

\usepackage{svn-multi}
\usepackage{amsmath}
\usepackage{amssymb}
\usepackage{amsthm}
\usepackage{mathtools}
\usepackage{bbm}
\usepackage{stmaryrd}
\usepackage{url}
\usepackage{hyperref}
\usepackage{yfonts}
\usepackage{rotating} % Rotating texts
\usepackage{multirow} % Merging rows and columns in tables
\usepackage{xspace}
\usepackage{clrscode}
\usepackage{tikz}
\usepackage{pgfplots}
\usepackage{comment}
\usepackage{cite}
\usepackage{graphicx}
\usepackage{caption}
\usepackage{subcaption}
\usepackage{ifthen}
\usepackage{paralist}
\usepackage{pgf}
\usetikzlibrary{automata,arrows,intersections,positioning,calc,shapes}
\usepackage[plain]{algorithm} %[boxed] = mit Kasten
\usepackage[noend]{algorithmic}
\algsetup{indent=2em}

\svnid{$Id$}

\newif\ifdraft\draftfalse

\newcommand{\Red}[2]{\textrm{Red}}
\newcommand{\bisim}{\sim}

\newcommand{\states} {{\ensuremath S}}
\newcommand\dist{\mathit{Dist}}
\usepackage{bigstrut}

\usepackage{paralist}

\newcommand{\mystackrel}[2]{%
  \mathrel{\vbox{\offinterlineskip\ialign{%
    \hfil##\hfil\cr
    $\scriptstyle#1$\cr
    \noalign{\kern.2ex}
    $#2$\cr
}}}}

 \newtheorem{theorem}{Theorem}
 \newtheorem{proposition}{Proposition}
 \newtheorem{corollary}{Corollary}

 \newtheorem{definition}{Definition}
\theoremstyle{definition}
\newtheorem{example}{Example}[section]

\newboolean{useComments}
\setboolean{useComments}{true}

\ifthenelse{\boolean{useComments}}{%
        \newcommand{\todo}[1]{\textcolor{olive}{{#1}}}
        \newcommand{\hh}[1]{\textcolor{teal}{\texttt{ Hassan: #1}}}
        \newcommand{\vh}[1]{{\color{red}\texttt{ Vahid: #1}}}
        \newcommand{\jk}[1]{{\color{blue}\texttt{ Jan: #1}}}
} { %
        \newcommand{\todo}[1]{}
        \newcommand{\hh}[1]{}
        \newcommand{\vh}[1]{}
        \newcommand{\jk}[1]{}
}

\newcommand{\bigO}[1]{\mathcal{O}(#1)}

\renewcommand{\epsilon}{\varepsilon}

\newcommand{\reals}{\mathbb{R}}
\newcommand{\posreals}{\reals^{\geq 0}}
\newcommand{\rationals}{\mathbb{Q}}

\newcommand{\setcond}[2]{\{\, #1 \mid #2 \,\}}

\newcommand{\setcardinality}[1]{|#1|}

\newcommand{\powerSet}[1]{2^{#1}}

\newcommand{\sd}{\mu} % states distribution
\newcommand{\gd}{\rho} % generic distribution

\newcommand{\functioneval}[2]{#1(#2)}

\newcommand{\Disc}[1]{\Delta(#1)}

\newcommand{\imdp}[1][M]{\mathfrak{#1}}
\newcommand{\stateSet}{S}
\newcommand{\actionSet}{{A}}

\newcommand{\APSet}{\mathit{AP}}
\newcommand{\APLabelling}{L}

\newcommand{\intTransitionProbability}{\mathit{I}}

\newcommand{\intervalSet}{\mathbb{I}}
\newcommand{\dotteddiamond}{\kern-1pt\setbox0=\vbox{\hbox{$\diamond$}}{\ooalign{\hfil\box0\hfil\cr\hfil$\mkern-0.5mu \cdot$\hfil\crcr}}\kern-1pt}

\newcommand{\LPproblemMatchingRelGdaBetaaGdbBetab}[5]{{#2 \mystackrel{#3}{\Longrightarrow}_{\combined} \dotteddiamond\liftrel[{#1}]\dotteddiamond{}\hspace{2pt} {}_{\kern4pt \combined}{\mystackrel{\mkern.5mu #5}{\Longleftarrow}}#4}}

\newcommand{\node}{\mathcal{V}}

\newcommand{\last}[1]{\mathit{last}(#1)}

\newcommand{\combined}{\mathrm{C}}

\newcommand{\relord}[1][\relsymbol]{#1}

\newcommand{\liftrel}[1][\relord]{\mathrel{\liftrelord[#1]}}
\newcommand{\liftrelord}[1][\relord]{\mathord{\mathcal{L}(#1)}}

\newcommand{\equivclass}[0]{\mathcal{C}}
\newcommand{\relclass}[2]{[#1]_{#2}}

\newcommand{\quotienting}[2]{#1/#2}
\newcommand{\partitionset}[2][\relord]{\quotienting{#2}{#1}}

\newcommand{\strongBisim}{\sim}

\newcommand{\acronym}[1]{\ensuremath{\textsl{#1}}}

\newcommand{\MDP}{\acronym{MDP}}

\newcommand{\IMDP}{\acronym{IMDP}}

\newcommand{\progHeader}[1]{#1}

\newcommand{\ComputeBisimProc}{\procedure{Bisimulation}}

\newcommand{\ViolateProc}{\procedure{Violate}}
\newcommand{\MinimalProc}{\procedure{Minimal}}

\newcommand{\partitioningord}{\relord[R]}

\def\vec#1{\mathbf{#1}}

\newcommand{\naturals}{\mathbb{N}}

\newcommand{\procedure}[1]{\texttt{#1}}

\newcommand{\J}{\Join}
\newcommand{\nxt}{{\text{\sf X}}}
\newcommand{\unt}{{\text{\sf U}}}
\newcommand{\pctlP}{{\text{\sf P}}}

\newcommand{\true}{{\textit{true}}}

\newcommand{\fpth}{\mathit{Paths}_{\mathit{fin}}}
\newcommand{\infpth}{\mathit{Paths}_{\mathit{inf}}}
\newcommand{\ifpth}[1]{\mathit{Paths}_{#1}}   
\newcommand{\Sch}{\Sigma}
\newcommand{\Sched}{\emph{Scheduler}}
\newcommand{\Envir}{\emph{Nature}}
\newcommand{\sched}{\emph{scheduler}}
\newcommand{\envir}{\emph{nature}}
\newcommand{\Env}{\Pi}
\newcommand{\env}{\pi}
\newcommand{\sch}{\sigma}

\newcommand{\uncertainty}[2]{\mathcal{E}^{#1,#2}}
\newcommand{\paruncertainty}[3]{\mathcal{E}^{#1,#2}_{#3}}
\newcommand{\parcombuncertainty}[2]{\mathcal{E}^{#1}_{#2}}

\newcommand{\distrib}[1]{\Disc{#1}}
\def\prob{\textnormal{Pr}}

\newcommand{\dis}{\sd}

\newcommand{\un}{{(\forall)}}
\newcommand{\ex}{{(\exists)}}
\newcommand{\controlSyn}{{(\exists\sch\forall)}}
\newcommand{\paramSyn}{{(\exists\env\forall)}}
\newcommand{\synthesis}{{(\exists\forall)}}

\newcommand{\pat}{\omega}
\newcommand{\sigmaalgebra}{\mathcal{B}}

\newcommand{\transition}[2]{{#1 \longrightarrow #2}}
\newcommand{\transitionAction}[3]{{#1 \mystackrel{#2}{\longrightarrow} #3}}

\newcommand{\eqclass}{\mathcal{C}}

\DeclareMathOperator*{\convexComb}{conv}

\newcommand{\tra}[1]{{}\mathchoice%
    {\stackrel{#1}{\longrightarrow}}
    {\mathop {\smash\longrightarrow}\limits^{\vrule width 0pt height 0pt depth 4pt\smash{#1}}}
    {\stackrel{#1}{\longrightarrow}}
    {\stackrel{#1}{\longrightarrow}}
{}}

\title{Probabilistic Bisimulations for PCTL Model Checking of Interval MDPs}
\author{Vahid Hashemi
\institute{MPI f\"{u}r Informatik\\ Saarbr\"{u}cken, Germany}
\institute{Department of Computer Science\\
Saarland University\\
Saarbr\"{u}cken, Germany}
\email{hashemi@mpi-inf.mpg.de}
\and
Hassan Hatefi
\institute{Department of Computer Science\\
Saarland University\\
Saarbr\"{u}cken, Germany}
\email{hhatefi@cs.uni-saarland.de}
\and
 Jan Kr\v{c}\'{a}l
\institute{Department of Computer Science\\
Saarland University\\
Saarbr\"{u}cken, Germany}
\email{krcal@cs.uni-saarland.de}
}
\def\titlerunning{Probabilistic Bisimulations for PCTL Model Checking of Interval MDPs}
\def\authorrunning{V. Hashemi, H. Hatefi \& J. Kr\v{c}\'{a}l}
\begin{document}
\maketitle
\begin{abstract}
%Verification of PCTL properties of MDPs with convex uncertainties has been investigated recently by Puggelli et al. However, model checking algorithms typically suffer from state space explosion. In this paper, we define strong probabilistic bisimulation to reduce the size of such an MDPs while preserving PCTL properties it satisfies. We also give a polynomial time algorithm to compute the bisimulation quotient. On the resulting behaviorally equivalent and minimized system, the existing PCTL model checking algorithm can be applied. Finally, we show the effectiveness of our reduction using a case study.
%
Verification of PCTL properties of MDPs with convex uncertainties has been investigated recently by Puggelli et al. However, model checking algorithms typically suffer from state space explosion. In this paper, we address probabilistic bisimulation to reduce the size of such an MDPs while preserving PCTL properties it satisfies. 
We discuss different interpretations of uncertainty in the models which are studied in the literature and that result in two different definitions of bisimulations.
We give algorithms to compute the quotients of these bisimulations in time polynomial in the size of the model and exponential in the uncertain branching. 
Finally, we show by a case study that large models in practice can have small branching and that a substantial state space reduction can be achieved by our approach.
%
%the effectiveness of our reduction using a case study.
\end{abstract}

\section{Introduction}
\label{sec:Introduction}

Modelling formalisms like Markov decision processes (MDP)\cite{Puterman:1994:MDP:528623} or Probabilistic automata (PA)~\cite{Segala-thesis} are used for representing systems that combine non-deterministic and probabilistic behaviour.
They can be viewed as transition systems where in each step an outgoing transition of the current state is chosen \emph{non-deterministically} and the successor state is chosen \emph{randomly} according to a fixed probability distribution assigned to this transition.
Assigning fixed probability distributions to transitions is however not realistic~\cite{DBLP:conf/lics/JonssonL91, DBLP:journals/rc/KozineU02} in many modelling scenarios: measurement errors, statistical estimates, or mathematical approximations all lead to \emph{intervals} instead of fixed probabilities.

\emph{Interval MDPs}~\cite{DBLP:conf/cav/PuggelliLSS13} (also called \emph{Bounded-parameter MDPs}~\cite{DBLP:journals/ai/GivanLD00,DBLP:journals/ai/WuK08}) address this need by
bounding the probabilities of each successor state by an interval instead of a fixed number.
In such a model, the transition probabilities are not fully specified and this uncertainty again needs to be resolved non-deterministically. 
The two sources of non-determinism have \emph{different} interpretation in different applications: 
\begin{enumerate}
\item In verification of parallel systems with uncertain transition probabilities~\cite{DBLP:conf/cav/PuggelliLSS13} the transitions correspond to unpredictable interleaving of computation of the communicating agents. Hence, both the choice of transitions and their probability distributions is \emph{adversarial}.
\item In control synthesis for systems with uncertain probabilities~\cite{DBLP:conf/cdc/WolffTM12} the transitions correspond to various control actions. We search for a choice of transitions that is \emph{optimal} against an adversarial choice of probability distributions satisfying the interval bounds.
\item In parameter synthesis for parallel systems~\cite{DBLP:conf/nfm/HahnHZ11} the transition probabilities are underspecified to allow freedom in implementation of such a model. We search for a choice of probability distributions that is optimal for adversarial choice of transitions (again stemming from the possible interleaving).
\end{enumerate}
Furthermore, the choice of probability distributions satisfying the interval constraints can be either resolved statically~\cite{DBLP:conf/lics/JonssonL91}, i.e. at the beginning once for all, or dynamically~\cite{DBLP:journals/mor/Iyengar05, DBLP:conf/tacas/SenVA06}, i.e. independently for each computation step. Here, we focus on the dynamic approach that is easier to work with algorithmically and can be seen as a relaxation of the static approach that is 
%in many contexts more faithful, yet 
often intractable~\cite{DBLP:conf/tacas/BenediktLW13,DBLP:conf/tacas/SenVA06,DBLP:conf/fossacs/ChatterjeeSH08,DBLP:conf/lata/DelahayeLLPW11}.
%
%Although being more faithful in many contexts, the static approach leads to much harder problems~\cite{}. Therefore, the dynamical approach is usually taken as a relaxation of the statical one~\cite{}. 

%co chci rict: Umime verifikovat. Snadno mame velky modely, ktery jsou trochu redundantni. Bisimulace je standardni nastroj pro redukci stavoveho prostoru - obzvlaste vhodna pro compositional modelling.
%
%When dealing with complicated systems, some higher-level modelling formalism such as a process algebra 
%
%In formal verification, some higher-level modelling formalism such as a process algebra is applied to obtain 
%
%A good framework for formal verification usually contains several ingredients: a formalism, such as interval MDP, for representing models, 

There are several algorithms~\cite{DBLP:conf/cdc/WolffTM12, DBLP:conf/cav/PuggelliLSS13} to check whether a given interval MDP satisfies a given specification expressed in a \emph{logic} like PCTL~\cite{DBLP:journals/fac/HanssonJ94} or LTL~\cite{DBLP:conf/focs/Pnueli77}. 
However, models often suffer from state space explosion when obtained using some higher-level modelling formalism
%~\cite{something-general-state-space-explosion} 
such as a process algebra. These models usually contain redundancy that can be removed without changing the behaviour of the model.
One way to reason about such behavioural equivalence is \emph{bisimulation}~\cite{DBLP:books/daglib/0067019}. For a given huge model it allows to construct the \emph{bisimulation quotient}, the smallest model with equivalent behaviour -- in particular preserving all its properties expressible by a logic such as PCTL.

%In this paper, we define the first bisimulations for the different interpretations of interval MDPs. (These bisimulations are also the first bisimulations for MDPs with uncertain transitions in general.) 
%%
%Furthermore, we show how to compute these bisimulations by algorithms based on comparing polytopes of probability distributions associated with each transition. We also discuss under what conditions a simpler symbolical approach can be applied instead of this geometrical one. 
%%
%Finally, using a case study we demonstrate possibilities for higher-level compositional modelling over interval MDPs. We show that the compositional modelling guarantees the conditions for the simpler symbolical approach if the uncertainty stems from a small number of different phenomena such as \emph{node failure} or \emph{loss of a message}. This limited source of uncertainty then may result in a massive state space reduction.

\paragraph{Our contribution}

In this paper, we define the first bisimulations for interval MDPs (that are also the first bisimulations for MDPs with uncertain transitions in general). We show that different interpretation of non-determinism yields two different bisimulations: one for models where the two non-determinisms are resolved in a \emph{cooperative way} (see point 1. above), another for models where it is resolved in a \emph{competitive way} (see points 2. and 3. above).
%

%%%%%%%%%%%%%%%% Explanation of FPT and PTIME algorithms
%Furthermore, we show how to compute these bisimulations by algorithms based on comparing polytopes of probability distributions associated with each transition. For the general setting, these algorithms are fixed parameter tractable with respect to the maximal dimension of the polytopes (i.e. maximal number of different states that an uncertain transition can lead to). 
%
%We achieve polynomial complexity for two settings.
%
%First, for systems where the transitions are chosen by \emph{deterministic scheduler}, the comparison of polytopes can be done much easier. Note that deterministic schedulers are assumed for the model-checking algorithm of~\cite{DBLP:conf/cav/PuggelliLSS13}.
%
%Second, roughly speaking, for systems with small number of different shapes of polytopes, a symbolical approach can be taken instead the geometrical one.
%

%
%We finally argue by a case study that this second polynomially solvable setting is naturally obtained by compositional modelling where the uncertainty stems from a small number of different phenomena such as \emph{node failure} or \emph{loss of a message}.
%
%The same shape of the polytope of one phenomenon then may repeat many times over different successor states due to the interleaving of execution. We demonstrate that this redundancy may result in a massive state space reduction.
%%%%%%%%%%%%%%% End of explanation of FPT and PTIME algorithm

Furthermore, we show how to compute these bisimulations by algorithms
based on comparing polytopes of probability distributions associated
with each transition. The algorithms are fixed parameter tractable
with respect to the maximal dimension of the polytopes (i.e. maximal
number of different states that an uncertain transition can lead to); in the competitive case also with respect to the maximal number of outgoing uncertain transitions.
Note that in many applications these parameters are small.
%, then we can hope that the
%problem for real applications is tractable.

We finally argue by a case study that, if uncertainty stems from a
small number of different phenomena such as \emph{node failure} or
\emph{loss of a message}, the same shape of polytopes will repeat many
times over the states space.  We demonstrate that the redundancy in
this case may result in a massive state space minimisation.

\begin{example} \label{ex:intro}
We illustrate the contribution by two examples. In the first one, we explain how the competitive and the cooperative resolution of non-determinism result in different behavioural equivalences. Consider the three pair of states below.

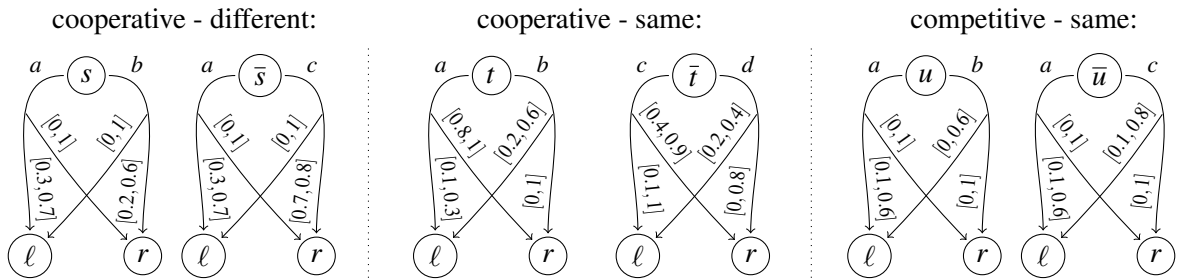
\begin{center}
\begin{tikzpicture}[x=2.5cm,y=1.2cm]

\node[above]  at (3em,0.2) {cooperative - different:};

\begin{scope}[xshift=14em]
\node[above]  at (3em,0.2) {cooperative - same:};
\end{scope}

\begin{scope}[xshift=29em]
\node[above]  at (3em,0.2) {competitive - same:};
\end{scope}

\end{tikzpicture}
\begin{tikzpicture}[x=2.5cm,y=1.2cm,outer sep=0.5mm,
state/.style={draw,circle, inner sep =0.25em,text centered},
trans/.style={font=\footnotesize},
prob/.style={font=\scriptsize,xscale=0.9}
]

\begin{scope}[]

\begin{scope}

\node[state] (top) at (0,0) {$s$};
\node[state] (left) at (-0.3,-2) {$\ell$};
\node[state] (right) at (0.3,-2) {$r$};

\path[->] (top) edge[in=90,out=180,looseness=1] node[above left=-4pt,pos=0.1,trans]{$a$} node[pos=0.4,name=topleft,inner sep=0,outer sep=0]{} node[above=-3pt,sloped,pos=0.8,prob] {$[0.3,0.7]$} (left);

\path[->] (top) edge[in=90,out=0,looseness=1] node[above right=-4pt,pos=0.1,trans]{$b$} node[pos=0.4,name=topright,inner sep=0,outer sep=0]{} node[above=-3pt,pos=0.8,prob,sloped] {$[0.2,0.6]$} (right);

\path[->] (topleft) edge[bend right=5] node[above=-3pt,pos=0.2,prob,sloped] {$[0,1]$} (right);
\path[->] (topright) edge[bend left=5] node[above=-3pt,pos=0.2,prob,sloped] {$[0,1]$} (left);
\end{scope}

\begin{scope}[xshift=6em]
\node[state] (top) at (0,0) {$\overline{s}$};
\node[state] (left) at (-0.3,-2) {$\ell$};
\node[state] (right) at (0.3,-2) {$r$};

\path[->] (top) edge[in=90,out=180,looseness=1] node[above left=-4pt,pos=0.1,trans]{$a$} node[pos=0.4,name=topleft,inner sep=0,outer sep=0]{} node[above=-3pt,sloped,pos=0.8,prob] {$[0.3,0.7]$} (left);

\path[->] (top) edge[in=90,out=0,looseness=1] node[above right=-4pt,pos=0.1,trans]{$c$} node[pos=0.4,name=topright,inner sep=0,outer sep=0]{} node[above=-3pt,pos=0.8,prob,sloped] {$[0.7,0.8]$} (right);

\path[->] (topleft) edge[bend right=5] node[above=-3pt,pos=0.2,prob,sloped] {$[0,1]$} (right);
\path[->] (topright) edge[bend left=5] node[above=-3pt,pos=0.2,prob,sloped] {$[0,1]$} (left);
\end{scope}

\draw [dotted] (1.5,0.3) -- (1.5,-2.3);

\end{scope}

\begin{scope}[xshift=14em]

\begin{scope}
\node[state] (top) at (0,0) {$t$};
\node[state] (left) at (-0.3,-2) {$\ell$};
\node[state] (right) at (0.3,-2) {$r$};

\path[->] (top) edge[in=90,out=180,looseness=1] node[above left=-4pt,pos=0.1,trans]{$a$} node[pos=0.4,name=topleft,inner sep=0,outer sep=0]{} node[above=-3pt,sloped,pos=0.8,prob] {$[0.1,0.3]$} (left);

\path[->] (top) edge[in=90,out=0,looseness=1] node[above right=-4pt,pos=0.1,trans]{$b$} node[pos=0.4,name=topright,inner sep=0,outer sep=0]{} node[above=-3pt,pos=0.8,prob,sloped] {$[0,1]$} (right);

\path[->] (topleft) edge[bend right=5] node[above=-3pt,pos=0.2,prob,sloped] {$[0.8,1]$} (right);
\path[->] (topright) edge[bend left=5] node[above=-3pt,pos=0.2,prob,sloped] {$[0.2,0.6]$} (left);
\end{scope}

\begin{scope}[xshift=7em]
\node[state] (top) at (0,0) {$\overline{t}$};
\node[state] (left) at (-0.3,-2) {$\ell$};
\node[state] (right) at (0.3,-2) {$r$};

\path[->] (top) edge[in=90,out=180,looseness=1] node[above left=-4pt,pos=0.1,trans]{$c$} node[pos=0.4,name=topleft,inner sep=0,outer sep=0]{} node[above=-3pt,sloped,pos=0.8,prob] {$[0.1,1]$} (left);

\path[->] (top) edge[in=90,out=0,looseness=1] node[above right=-4pt,pos=0.1,trans]{$d$} node[pos=0.4,name=topright,inner sep=0,outer sep=0]{} node[above=-3pt,pos=0.8,prob,sloped] {$[0,0.8]$} (right);

\path[->] (topleft) edge[bend right=5] node[above=-3pt,pos=0.2,prob,sloped] {$[0.4,0.9]$} (right);
\path[->] (topright) edge[bend left=5] node[above=-3pt,pos=0.2,prob,sloped] {$[0.2,0.4]$} (left);
\end{scope}

\draw [dotted] (1.7,0.3) -- (1.7,-2.3);

\end{scope}

\begin{scope}[xshift=29em]

\begin{scope}
\node[state] (top) at (0,0) {$u$};
\node[state] (left) at (-0.3,-2) {$\ell$};
\node[state] (right) at (0.3,-2) {$r$};

\path[->] (top) edge[in=90,out=180,looseness=1] node[above left=-4pt,pos=0.1,trans]{$a$} node[pos=0.4,name=topleft,inner sep=0,outer sep=0]{} node[above=-3pt,sloped,pos=0.8,prob] {$[0.1,0.6]$} (left);

\path[->] (top) edge[in=90,out=0,looseness=1] node[above right=-4pt,pos=0.1,trans]{$b$} node[pos=0.4,name=topright,inner sep=0,outer sep=0]{} node[above=-3pt,pos=0.8,prob,sloped] {$[0,1]$} (right);

\path[->] (topleft) edge[bend right=5] node[above=-3pt,pos=0.2,prob,sloped] {$[0,1]$} (right);
\path[->] (topright) edge[bend left=5] node[above=-3pt,pos=0.2,prob,sloped] {$[0,0.6]$} (left);
\end{scope}

\begin{scope}[xshift=6em]
\node[state] (top) at (0,0) {$\overline{u}$};
\node[state] (left) at (-0.3,-2) {$\ell$};
\node[state] (right) at (0.3,-2) {$r$};

\path[->] (top) edge[in=90,out=180,looseness=1] node[above left=-4pt,pos=0.1,trans]{$a$} node[pos=0.4,name=topleft,inner sep=0,outer sep=0]{} node[above=-3pt,sloped,pos=0.8,prob] {$[0.1,0.6]$} (left);

\path[->] (top) edge[in=90,out=0,looseness=1] node[above right=-4pt,pos=0.1,trans]{$c$} node[pos=0.4,name=topright,inner sep=0,outer sep=0]{} node[above=-3pt,pos=0.8,prob,sloped] {$[0,1]$} (right);

\path[->] (topleft) edge[bend right=5] node[above=-3pt,pos=0.2,prob,sloped] {$[0,1]$} (right);
\path[->] (topright) edge[bend left=5] node[above=-3pt,pos=0.2,prob,sloped] {$[0.1,0.8]$} (left);
\end{scope}

% \draw [dotted] (1.55,0.3) -- (1.55,-2.3);

\end{scope}

\end{tikzpicture}
\end{center} 
\end{example}
As regards the cooperative non-determinism, $s$ has not the same behaviour as $\overline{s}$ since $\overline{s}$ can move to $r$ with probability $0.8$ by choosing $c$ and $(\ell \mapsto 0.2, r \mapsto 0.8)$, which ${s}$ cannot simulate. So far the equivalence might seem easy to check. However, note that $t$ has the same behaviour as $\overline{t}$ even though the interval bounds for the transitions quite differ. 
%Below, we sketch how to decide this more complicated case. 
Indeed, the sets of distributions satisfying the interval constraints are the same for $t$ and $\overline{t}$.

As regards the competitive non-determinism, observe that $u$ and $\overline{u}$ have also the same behaviour.
Indeed, the $a$ transitions coincide and both $b$ and $c$ offer a wider choice of probability distributions than $a$. If the most adversarial choice of the distribution scheduler lies in the difference $[b] \setminus [a]$ of the distributions offered by $b$ and $a$, the transition scheduler then never chooses $b$; hence $a$ in $\overline{u}$ can simulate both $a$ and $b$ in $u$. In the other direction it is similar and $u$ and $\overline{u}$ have the same behaviour although $[b] \neq [c]$.

\begin{example} In the second example, we explain the 
%substantial 
redundancy of large models with a small source of uncertainty. 
Consider a Wireless Sensor Network (WSN)
containing $N$ sensors $S_1,S_2\cdots S_N$ and a gateway $G$, all
communicating over an unreliable channel.  For simplicity, we assume
that each sensor continuously sends some data to the gateway which are
then pushed into an external server for further analysis.  As the
channel is unreliable, with some positive probability $p$ each message with data may get lost.
%, in this case the sensor goes to fail state.
%We consider a simple model of the gateway containing only one state
%that successively receives the data from sensors as depicted in
The WSN can be seen as the parallel composition
of gateway $G$ and sensors $S_i$ 
%from Figure~\ref{fig:wsn} 
depicted below
that synchronise over labels \texttt{send}$_i$'s and \texttt{receive}$_i$'s.

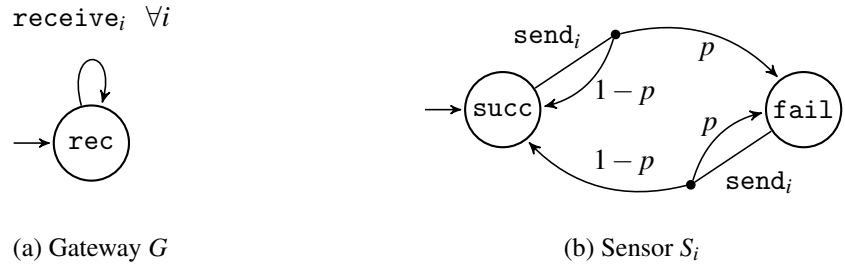
\vspace{-2em}
\begin{figure}[h]
  \centering
  \begin{subfigure}[b]{0.45\textwidth}
    \centering
    \begin{tikzpicture}[shorten >=1pt,auto,>=stealth',semithick,baseline]
%      \draw [white,help lines] (0,0) grid (6,2);
      \tikzstyle{every node}=[circle, thick, inner sep = 0pt]
      \node [draw, minimum size = 10mm, initial, initial text={}] (0) at (8,1) {\texttt{rec}};
      \path
      (0) edge [loop above] node [above = -5mm] {\texttt{receive}$_i \;\;\; \forall i$} (0);
    \end{tikzpicture}
    \caption{Gateway $G$}
  \end{subfigure}%
  \begin{subfigure}[b]{0.45\textwidth}
    \centering
    \begin{tikzpicture}[shorten >=1pt,auto,>=stealth',semithick,baseline]
%      \draw [,white,help lines] (0,0) grid (6,2);
      \tikzstyle{every node}=[circle, thick, inner sep = 0pt]
      \node [draw, minimum size = 10mm, initial, initial text={}] (0) at(2,1){\texttt{succ}};
      \node [draw,minimum size = 10mm] (1) at (6,1) {\texttt{fail}};
      \node [draw, fill, minimum size = 1mm] (00) at (3.5,2) {};
      \node [draw, fill, minimum size = 1mm] (11) at (4.5,0) {};
      \path
      (0) edge [shorten >= 0] node {\texttt{send}$_i$} (00)
      (00)edge [->,bend left=30] node[below=2pt] {$p$}              (1)
          edge [->,bend left=30] node[pos=0.3] {$1-p$}                (0) 
      (1) edge [shorten >= 0] node {\texttt{send}$_i$} (11)
      (11)edge [->,bend left=30] node {$p$}              (1)
          edge [->,bend left=30] node[swap] {$1-p$}                (0) 
      ;
    \end{tikzpicture}
    \caption{Sensor $S_i$}
  \end{subfigure}%
%  \caption{Models of sensors and the gateway}
  \label{fig:wsn}
\end{figure}
For
instance environmental effects on radio transmission, mobility of
sensor nodes or traffic burst (see
e.~g.~\cite{mringwal:phdthesis:2009}) cause that the exact probability of
failure is unknown. 
% and therefore can be estimated via e.~g. empirical
%data analysis.  
The estimation of this probability, 
%for instance 
e.g. 
by empirical data analysis, usually leads to an 
interval $p \in [\ell, u]$ which turns the model into an interval MDP.  

Let us stress that 
%for arbitrarily large number of sensors $N$, 
%The key point
%of the model is that 
there is only one source of uncertainty appearing
all over the state space no matter what is the number of sensors $N$.
This makes many states of the model behave
similarly.  For example in the WSN, the parallel composition of the
above model has $2^N$ states.  However one can show that the
bisimulation quotient has only $N+1$ states. Indeed, all states that have the same number of failed sensors have the same behaviour.  
Thus, for limited source of uncertainty in a model obtained by compositional modelling, the state space reduction may be enormous.
%The main point here is that if we have a very limited sources of uncertainty exhibiting in a model, there are likely many states of the model that behave similarly.

\end{example}

\paragraph{Related work}
Various probabilistic modelling formalisms with uncertain transitions are studied in the literature. Interval Markov chains~\cite{DBLP:conf/lics/JonssonL91,DBLP:journals/rc/KozineU02} or Abstract Markov chains~\cite{DBLP:conf/spin/FecherLW06} extend standard discrete-time Markov chains (MC) with interval uncertainties and thus do not feature the non-deterministic choice of transitions. Uncertain MDPs~\cite{DBLP:journals/ior/NilimG05,DBLP:conf/cdc/WolffTM12,DBLP:conf/cav/PuggelliLSS13} allow more general sets of distributions to be associated with each transition, not only those described by intervals. Usually, they restrict to \emph{rectangular uncertainty sets} requiring that the uncertainty is linear and independent for any two transitions of any two states. Our general algorithm working with polytopes can be easily adapted to this setting. Parametric MDPs~\cite{DBLP:conf/nfm/HahnHZ11} to the contrary allow such dependencies as every probability is described as a rational function of a finite set of global parameters.

From the side of view of compositional specification, Interval Markov chains~\cite{DBLP:conf/lics/JonssonL91} and Abstract probabilistic automata~\cite{DBLP:conf/vmcai/DelahayeKLLPSW11, DBLP:conf/acsd/DelahayeKLLPSW11} serve as specification theories for MC and PA featuring satisfaction relation, and various refinement relations. In order to be closed under parallel composition, Abstract PA allow general polynomial constraints on probabilities instead of interval bounds. Since for Interval MC it is not possible to explicitly construct parallel composition, the problem whether there is a common implementation of a set of Interval Markov chains is addressed instead~\cite{DBLP:conf/lata/DelahayeLLPW11}. To the contrary, interval bounds on \emph{rates} of outgoing transitions work well with parallel composition in the continuous-time setting of Abstract interactive Markov chains~\cite{DBLP:conf/formats/KatoenKN09}. The reason is that unlike probabilities, the rates do not need to sum up to $1$. A different way~\cite{yi1994reasoning} to successfully define parallel composition for interval models is to separate synchronising transitions from the transitions with uncertain probabilities. This is also the core of our approach to parallel composition when constructing a case study as discussed in Section~\ref{sec:case-studies}.

We are not aware of any existing bisimulation for uncertain or parametric probabilistic models. Among similar concepts studied in the literature are simulation~\cite{yi1994reasoning} and refinement~\cite{DBLP:conf/lics/JonssonL91, DBLP:conf/lata/DelahayeLLPW11,DBLP:conf/vmcai/DelahayeKLLPSW11} relations for previously mentioned models. Our definition of bisimulation in the competitive setting is inspired by the alternating bisimulation~\cite{DBLP:conf/concur/AlurHKV98,DBLP:conf/csl/ChatterjeeCK12}. 

Many new verification algorithms for interval models appeared in last few years. 
Reachability and expected total reward is addressed for for Interval MC~\cite{DBLP:journals/ipl/ChenHK13} as well as Interval MDP~\cite{DBLP:journals/ai/WuK08}. 
PCTL model checking and LTL model checking are studied for Interval MC~\cite{DBLP:conf/fossacs/ChatterjeeSH08, DBLP:journals/ipl/ChenHK13,DBLP:conf/tacas/BenediktLW13} and also for Interval MDP~\cite{DBLP:conf/cav/PuggelliLSS13,DBLP:conf/cdc/WolffTM12}. 
Among other technical tools, all these approaches make use of (robust) dynamic programming relying on the fact that transition probability distributions are resolved dynamically. For the static resolution of distributions, adaptive discretisation technique for PCTL parameter synthesis is given in~\cite{DBLP:conf/nfm/HahnHZ11}. 
Uncertain models are also widely studied in the control community~\cite{DBLP:journals/ai/GivanLD00, DBLP:journals/ior/NilimG05,DBLP:journals/ai/WuK08}, mainly interested in maximal expected finite-horizon reward or maximal expected discounted reward.
%\jk{I should add more references to control community here}

 \paragraph{Structure of the paper}
 We start with necessary preliminaries in
 Section~\ref{sec:preliminaries}. In Section~\ref{sec:bisimulation},
 we give the definitions of probabilistic bisimulations for
 interval MDP and discuss their properties and differences. In
 Section~\ref{sec:deciding_strong_bisim}, we give the FPT algorithms
 for both cooperative and competitive cases. Finally, in
 Section~\ref{sec:case-studies} we demonstrate our approach on a case
 study. Due to space limitations, we refer the reader interested in detailed proofs to~\cite{HHK14}.

\section{Preliminaries}
\label{sec:preliminaries}

%\subsection{Functions and Relations}
%For a given set $X$ and $\bot \notin X$, we denote by $\boted{X}$ the set $X \cup \setnocond{\bot}$.  

In this paper, the sets of all positive integers, rational numbers, real numbers and non-negative real numbers are denoted by 
$\naturals$, $\rationals$, $\reals$, and $\posreals$, respectively. 
For a set $X$, we denote by $\Disc{X}$ the set of discrete probability distributions over $X$.

\subsection{Interval Markov Decision Processes}
Let us formally define Interval \MDP{}.
%We now consider an extension of \MDP{}s with \emph{intervals}, that is, the probability assigned by transition probability distributions to states are not fixed numbers; rather, they are known to lie within a given interval. Formally, 

\begin{definition}[\IMDP]\label{def:imdp}
	An \emph{Interval Markov Decision Process} (\IMDP{}) $\imdp$ is a tuple $(\stateSet, \actionSet, \APSet,\APLabelling, \intTransitionProbability)$, where $\stateSet$ is a finite set of \emph{states}, $\actionSet$ is a finite set of \emph{actions}, $\APSet$ is a finite set of \emph{atomic propositions}, $\APLabelling \colon \stateSet \to \powerSet{\APSet}$ is a \emph{labelling function}, and $\intTransitionProbability \colon \stateSet \times \actionSet \times \stateSet \to \intervalSet$ is an \emph{interval transition probability function} where $\intervalSet$ is a set of subintervals of $[0,1]$.
\end{definition}
\noindent
Furthermore, for each state $s$ and action $a$, we denote by 
$s \tra{a} \dis$ that $\dis\in \distrib{\states}$ is a \emph{feasible distribution}, i.e. for each state $s'$ we have $\dis(s') \in \intTransitionProbability(s,a,s')$.
%By $\uncertainty{s}{a} = $, we denote the set of feasible distributions for $s$ and $a$. 
%denote by $\uncertainty{s}{a}$ the set of \emph{feasible distributions} $\dis \in \distrib{\states}$ such that for each state $s'$ we have $\dis(s') \in \intTransitionProbability(s,a,s')$. We also write  to denote that $\mu \in \uncertainty{s}{a}$.
%
We require that the set $\{\dis \mid s \tra{a} \dis \}$, also denoted by $\uncertainty{s}{a}$, is non-empty for each state $s$ and action $a$.

An interval MDP is initiated in some state $s_1$ and then moves in discrete steps from state to state forming an infinite path $s_1 \, s_2 \, s_3 \cdots$. One step, say from state $s_i$, is performed as follows. First, an action $a \in \actionSet$ is chosen non-deterministically by {\Sched}. Then, {\Envir} resolves the uncertainty and chooses non-deterministically one corresponding feasible distribution $\dis \in \uncertainty{s_i}{a}$. Finally, the next state $s_{i+1}$ is chosen randomly according to the distribution $\dis$.

Let us define the semantics of an \IMDP{} formally.
%Let us formalize the non-deterministic choices of $\Sched$ and $\Envir$.
%
A \emph{path} is a finite or infinite sequence of states $\pat = s_{1} \, s_{2} \cdots$. 
%For a path $\pat$ of length $\geq i$, we denote by $\pat[i]$ the $i$-th state of $\pat$; 
For a finite path $\pat$, we denote by $\textit{last}(\pat)$ the last state of $\pat$.
The set of all finite paths 
%of $\imdp$ 
and the set of all infinite paths 
%of $\imdp$ 
are denoted by $\fpth$ and $\infpth$, respectively. Furthermore, let $\ifpth{\pat} = \{\pat\pat' \mid \pat'\in\infpth \}$ denote the set of paths that have the finite prefix $\pat \in \fpth$.

%
%For a path $\pat$ of length $\geq i$, we denote by $\pat[i]$ the $i$-th state of $\pat$, 
%For a finite path $\pat$, we denote by $\textit{last}(\pat)$ the last state of $\pat$. 
%
%and $\mu_{i}= \elementFixedUncertainty{s_{i}}{s_{i+1}}{a_{i}}$ where ${a_{i}}\in \stateActionSet{(s_{i})}$ and $\elementFixedUncertainty{s_{i}}{s_{i+1}}{a_{i}}>0$. 
%Furthermore, we denote by $\fpth$ the set of all finite paths of $\imdp$ and by 
%$\infpth$ the set of all infinite paths of $\imdp$.
%

% and $\ifpth{s}=\{\rho| \rho[0]= s \}$ the set of all paths started by state $s$.  
%For every path $\rho \in \fpth$, we denote by $\textit{last}({\rho})$ the last state in $\rho$. 

%In order to check quantitative properties of the \IMDP{}s, we need to specify a probability space over infinite paths.\vh{cite paper} In the setting of our model, however, two sources of nondeterminism and uncertainty need to be resolved in order to build a probability space. In each state of \IMDP{} $\imdp$, we use \textit{scheduler} that resolves the nondeterminism by choosing one of the available actions at that state. Formally,

\begin{definition}[{\Sched} and \Envir]{\label{def:scheduler}}
A {\sched} 
is a function $\sch : \fpth \to \distrib{\actionSet}$ that to each finite path $\pat$ assigns a distribution over the set of actions. 
A {\envir} is a function $\env : \fpth \times \actionSet \to \distrib{\stateSet}$ that to each finite path $\pat$ and action $a$ assigns a feasible distribution, i.e. an element of $\uncertainty{s}{a}$ where $s = \last{\pat}$.
We denote by $\Sch$ the set of all schedulers and by $\Env$ the set of all natures.
\end{definition}
\noindent
For an initial state $s$, a scheduler $\sch$, and a nature $\env$, let $\prob^{\sch,\env}_s$ denote the unique probability measure over $(\infpth, \sigmaalgebra)$\footnote{
Here, $\sigmaalgebra$ is the standard $\sigma$-algebra over $\infpth$ generated from the set of all cylinder sets $\{\ifpth{\pat} \mid \pat \in\fpth \}$. The unique probability measure is obtained by the application of the by extension theorem (see, e.g.~\cite{Billingsley1979}).
%\footnote{A \emph{$\sigma$-algebra} over a set $\Omega$ is a set $\sigmaalgebra \subseteq 2^{\Omega}$ that includes $\Omega$ and is closed under complement and countable union. A $\sigma$-algebra generated from a collection $X$ is the smallest $\sigma$-algebra containing $X$.
%
} such that the probability $\prob^{\sch,\env}_s[\ifpth{s'}]$ of starting in $s'$ equals $1$ if $s = s'$ and $0$, otherwise; and the probability $\prob^{\sch,\env}_s[\ifpth{\pat s'}]$ of traversing a finite path $\pat s'$ equals $\prob^{\sch,\env}_s[\ifpth{\pat}] \cdot \sum_{a\in\actionSet} \sch(\pat)(a) \cdot \env(\pat,a)(s')$.

%Observe that the non-determinism is resolved based on the whole path that the process traversed so far. 
%
%A scheduler $\sch$ is said to be \emph{memoryless} if $\sch(\pat) = \sch(\pat')$ for all finite paths with $\last{\pat} = \last{\pat'}$. Similarly, a nature $\env$ is said to be \emph{memoryless} if $\env(\pat,a) = \env(\pat',a)$ for all finite paths with $\last{\pat} = \last{\pat'}$ and all actions $a$.

%Furthermore, 
Observe that the scheduler does not choose an action but a \emph{distribution} over actions. It is well-known~\cite{Segala-thesis} that such randomisation brings more power in the context of bisimulations. 
To the contrary, nature is not allowed to randomise over the set of feasible distributions $\uncertainty{s}{a}$. This is in fact not necessary, since the set $\uncertainty{s}{a}$ is closed under convex combinations.
Finally, a scheduler $\sch$ is said to be \emph{deterministic} if $\sch(\pat)(a)=1$ for some action $a$ for all finite paths $\pat$.

\subsection{Probabilistic Computation Tree Logic (PCTL)}
\label{sec:pctl}

There are various ways how to describe properties of interval \MDP{}s. Here we focus on \emph{probabilistic CTL} (PCTL)~\cite{DBLP:journals/fac/HanssonJ94}.
The syntax of PCTL state formulas $\varphi$ and PCTL path formulas $\psi$ is given by:
\begin{align*}
\varphi &:= \true \mid x  \mid \neg\varphi \mid \varphi_1\wedge \varphi_2 \mid
\pctlP_{\J p}(\psi) \\
\psi &:= \nxt \varphi \mid \varphi_1\unt\varphi_2 \mid \varphi_1\unt^{\leq k}\varphi_2 \notag
\end{align*}
where $x\in AP$, ${p}\in [0,1]$ is a rational constant, $\J\in\{\leq, <,\geq, > \}$, and $k\in \naturals$.

The satisfaction relation for PCTL formulae depends on the way how non-determinism is resolved for the probabilistic operator $\pctlP_{\J p}(\psi)$. 
%For the common (cooperative) universally quantified setting, we \vh{Here common (cooperative) universally quantified is not defined. We need to modify it to not confuse the reader.}
When quantifying both the non-determinisms universally, we
% 
% For a state $s$ and a state formula $\varphi$, we
define the satisfaction relation $s \models_\un \varphi$ as follows: 
$s \models_\un x$ if $x \in \APLabelling(s)$; $s \models_\un \neg\varphi$ if not $s \models_\un \varphi$; $s \models_\un \varphi_1 \land \varphi_2$ if both $s \models_\un \varphi_1$ and $s \models_\un \varphi_2$; and
% . The satisfaction of the probabilistic operator $\pctlP_{\J p} (\psi)$ depends on the way, how non-determinism is resolved. For the universally quantified setting, we have
%
\begin{align*}
% s \;\models_\un\; \pctlP_{\J p} (\psi) 
%\quad\mbox{ if }\quad
%\prob^{\sch,\env}_s (\{\pat\in \infpth \mid \pat \models_\un \psi\}) \J p 
%\quad \forall \sch\in\Sch \; \forall \env\in\Env
 s \;\models_\un\; \pctlP_{\J p} (\psi) 
\quad\mbox{ if }\quad 
\forall \sch\in\Sch \;\; \forall \env\in\Env: \quad
\prob^{\sch,\env}_s \left[\models_\un \psi \right] \J p.
\tag*{$\un$}
\end{align*}
where $\models_\un \psi$ denotes the set of infinite paths $\{\pat\in \infpth \mid \pat \models_\un \psi\}$ and the satisfaction relation $\pat \models_\un \psi$ for an infinite path $\pat = s_1 s_2 \cdots$ and a path formula $\psi$ is given by:
\begin{align*}
\pat &\models_\un \;\; \nxt\varphi
&& \text{if} \quad 
\text{$s_2 \models \varphi$;} \\
\pat &\models_\un \;\; \varphi_1\unt^{\leq k}\varphi_2 
&& \text{if} \quad 
\text{there exists $i\leq k$ such that $s_i \models_\un \varphi_2$,} \\
&&& \quad \quad
\text{and $s_j \models_\un \varphi_1$ for every $0\le j<i$;} \\
\pat &\models_\un \;\; \varphi_1\unt\varphi_2 
&& \text{if} \quad 
\text{there exists $k\in\naturals$ such that $\pat \models_\un \varphi_1\unt^{\leq k}\varphi_2$.}
\end{align*}

\noindent
It is easy to show that the set $\models_\un \psi$ is measurable for any path formula $\psi$, hence the definition is correct.
We explain how the semantics differs for different resolution of non-determinism in the next section.

\section{Probabilistic Bimulations for Interval Markov decision processes}
\label{sec:bisimulation}

% For an equivalence relation $\rel \; \subseteq X \times X$, the \emph{lifting} of $R$ is a relation $\lifting{\rel} \;\subseteq\; \distrib{X} \times \distrib{X}$ such that $\dis \lifting{\rel} \distwo$ if for each equivalence class $\equivclass \in \partitionset{X}$, $\dis(\equivclass) = \distwo(\equivclass)$.
% 
% ---
Let us fix an interval MDP $(\stateSet, \actionSet, \APSet,\APLabelling, \intTransitionProbability)$.
In this section, we define probabilistic bisimulations for different interpretations of Interval MDP. Namely the bisimulation $\bisim_\un$ for the cooperative setting and bisimulations $\bisim_\controlSyn$ and $\bisim_\paramSyn$ for two different applications for the competitive setting. We then show that $\bisim_\controlSyn$ and $\bisim_\paramSyn$ actually coincide.

\subsection{Cooperative resolution of non-determinism}

In the context of verification of parallel systems with uncertain transition probabilities, it makes sense to assume that $\Sched$ and $\Envir$ are resolved \emph{cooperatively} in the most \emph{adversarial} way. This setting yields a bisimulation quite similar to standard probabilistic bisimulation for models with one type of non-determinism~\cite{DBLP:journals/iandc/LarsenS91}.
%\jk{no, I need a reference to state labelled bisim}
%
First, let us denote by $s\tra{} \dis$ that a transition from $s$ according to $\dis$ can be taken cooperatively, i.e. that there is a decision $\gd \in \dist(\actionSet)$ of $\Sched$ and decisions $s\tra{a} \dis_a$ of $\Envir$ for each $a$ such that $\dis = \sum_{a\in\actionSet} \gd(a) \cdot \dis_a$. In other words, $s\tra{} \dis$ if $\dis \in \convexComb \{\uncertainty{s}{a} \mid a\in\actionSet\}$ where $\convexComb X$ denotes the convex hull of $X$.

\begin{definition}\label{def:bisim-cooperative}
Let $R \subseteq \stateSet \times \stateSet$ be a equivalence relation. We say that $R$ is \emph{probabilistic $\un$-bisimulation} if for any $(s,t) \in R$ we have that $\APLabelling(s) = \APLabelling(t)$ and
\begin{align*}
\text{for each} \;\;\;& \transition{s}{\mu} \\
\text{there is} \;\;\;& \transition{t}{\nu} \;\;\; \text{such that} \;\; \text{$\mu(\equivclass) =  \nu(\equivclass)$ for each equivalence class $\equivclass \in \partitionset{\stateSet}$}.
\end{align*}
% 
% 
% for each $\transition{s}{\mu}$ there is a $\transition{t}{\nu}$ such that $\mu(\equivclass) =  \nu(\equivclass)$ for each equivalence class $\equivclass \in \partitionset{\stateSet}$.
% % 
Furthermore, we write $s \strongBisim_\un t$ if there is a probabilistic $\un$-bisimulation $R$ such that $(s,t) \in R$.
\end{definition}

Intuitively, each (cooperative) step of $\Sched$ and $\Envir$ from state $s$ needs to be matched by a (cooperative) step of $\Sched$ and $\Envir$ from state $t$; symmetrically, $s$ also needs to match $t$.
As a first result, we show that the bisimulation $\bisim_\un$ preserves the (cooperative) universally quantified PCTL satisfaction $\models_\un$. 
%\vh{this theorem states soundness of the definition. Should not be remarked here? Why it is not complete? }

\begin{theorem}\label{thm:coop-bisim}
For states $s \strongBisim_\un t$ and any PCTL formula $\varphi$, we have $s \models_\un \varphi$ if and only if $t \models_\un \varphi$.
\end{theorem}
%\begin{proof}
%\jk{Needs to be done.}. By structural induction. The only tricky part is the operator $\pctlP_{\J p} (\psi)$; again by structural induction for the path formulae -- the extremal probabilities of satisfying path formulae are the same by bisimilarity.
%\end{proof}

Dually, the non-determinism could also be resolved \emph{existentially}. This corresponds to the setting where we want to synthesise both the scheduler $\sch$ that controls the system and choice of feasible probability distributions $\env$ such that $\sch$ and $\env$ together guarantee a specified behaviour $\varphi$. This setting is formalised by the satisfaction relation $\models_\exists$ which is defined like $\models_\un$ except for the operator $\pctlP_{\J p} (\psi)$ where we set
\begin{align*}
 s \;\models_\ex\; \pctlP_{\J p} (\psi) 
\quad\mbox{ if }\quad 
\exists \sch\in\Sch \;\; \exists \env\in\Env: \quad
\prob^{\sch,\env}_s \left[ \models_\ex \psi \right] \J p.
\tag*{$\ex$}
\end{align*}
Note that for any formula of the form $\pctlP_{< p} (\psi)$, we have $s \models_\exists \pctlP_{< p} (\psi)$ if and only if we have $s \models_\un \neg \pctlP_{\geq p} (\psi)$. This can be easily generalised: for each state formula $\varphi$ we obtain a state formula $\overline{\varphi}$ such that $s \models_\exists \varphi$ if and only if $s \models_\un \overline{\varphi}$ for each state $s$. Hence $\bisim_\un$ also preserves $\models_\exists$.
\begin{corollary}\label{cor:existential}
 For states $s \strongBisim_\un t$ and any PCTL formula $\varphi$, we have $s \models_\ex \varphi$ if and only if $t \models_\exists \varphi$.
\end{corollary}

\subsection{Competitive resolution of non-determinism}

As already argued for in Section~\ref{sec:Introduction}, there are applications where it is natural to interpret the two sources of non-determinism in a competitive way.

\paragraph{Control synthesis under uncertainty}

In this setting we search for a scheduler $\sigma$ such that for any nature $\pi$, a fixed property $\varphi$ is satisfied. This corresponds to the satisfaction relation $\models_\controlSyn$, obtained similarly from $\models_\un$ by replacing the rule $\un$ with
%is defined like $\models_\un$ except for the operator $\pctlP_{\J p} (\psi)$ where we set
\begin{align*}
 s \;\models_\controlSyn\; \pctlP_{\J p} (\psi) 
\quad\mbox{ if }\quad 
\exists \sch\in\Sch \;\; \forall \env\in\Env: \quad
\prob^{\sch,\env}_s \left[ \models_\controlSyn \psi \right] \J p.
\tag*{$\controlSyn$}
\end{align*}

As regards bisimulation, the competitive setting is not a common one. We define a bisimulation similar to the alternating bisimulation of~\cite{DBLP:conf/concur/AlurHKV98} applied to non-stochastic two-player games. For a decision $\gd \in \distrib{\actionSet}$ of \Sched, let us denote by $s \tra{\gd} \mu$ that $\mu$ is a possible successor distribution, i.e. there are decisions $\mu_a$ of {\Envir} for each $a$ such that $\mu = \sum_{a\in\actionSet} \gd(a) \cdot \mu_a$. 
\begin{definition}\label{def:bisim-control-syn}
Let $R \subseteq \stateSet \times \stateSet$ be an equivalence relation. We say that $R$ is \emph{probabilistic $\controlSyn$-bisimulation} if for any $(s,t) \in R$ we have that $\APLabelling(s) = \APLabelling(t)$ and
\begin{align*}
\text{for each} \;\;\;& \gd_s \in \distrib{\actionSet} \\
\text{there is} \;\;\;& \gd_t \in \distrib{\actionSet}\\
\text{such that for each} \;\;\;& \transitionAction{t}{\gd_t}{\nu} \\
\text{there is} \;\;\;& \transitionAction{s}{\gd_s}{\mu} \;\;\; \text{such that} \;\; \text{$\mu(\equivclass) =  \nu(\equivclass)$ for each equivalence class $\equivclass \in \partitionset{\stateSet}$}.
\end{align*}
Furthermore, we write $s \strongBisim_\controlSyn t$ if there is a probabilistic $\controlSyn$-bisimulation $R$ such that $(s,t) \in R$.
\end{definition}
The exact alternation of quantifiers might be counter-intuitive at first sight. Note that it exactly corresponds to the situation in non-stochastic games~\cite{DBLP:conf/concur/AlurHKV98} and that this bisimulation preserves the PCTL logic with $\models_\controlSyn$.

\begin{theorem}\label{thm:preserve-control-syn}
For states $s \strongBisim_\controlSyn t$ and any PCTL formula $\varphi$, we have $s \models_\controlSyn \varphi$ if and only if $t \models_\controlSyn \varphi$.
\end{theorem}
%\begin{proof}
%\jk{TODO. This proof needs mainly to give intuition why the strange alternation of quantifiers correspond to the satisfaction relation.}.
%\end{proof}

Similarly to Corollary~\ref{cor:existential}, we could define a satisfaction relation with the alternation $\forall \sch\in\Sch \; \exists \env\in\Env$ that is then preserved by the same bisimulation $\bisim_\controlSyn$. However, we see no natural application thereof.

\paragraph{Parameter synthesis in parallel systems}

In this setting, we search for a resolution $\pi$ of the underspecified probabilities such that for any scheduler $\sigma$ resolving the interleaving non-determinism, a fixed property $\varphi$ is satisfied. This corresponds to the satisfaction relation $\models_\paramSyn$, obtained similarly from $\models_\un$ by replacing the rule $\un$ with
%is defined like $\models_\un$ except for the operator $\pctlP_{\J p} (\psi)$ where we set
\begin{align*}
 s \;\models_\paramSyn\; \pctlP_{\J p} (\psi) 
\quad\mbox{ if }\quad 
\exists \env\in\Env \;\; \forall \sch\in\Sch: \quad
\prob^{\sch,\env}_s \left[ \models_\paramSyn \psi \right] \J p.
\tag*{$\paramSyn$}
\end{align*}
This yields a definition of bisimulation similar to Definition~\ref{def:bisim-control-syn}. For a choice $(\mu_a)_{a\in\actionSet}$ of underspecified probabilities, let us denote by $s \tra{(\mu_a)\;\;} \mu$ that $\mu$ is a possible successor distribution, i.e. there is a decision $\gd$ of {\Sched} such that 
$\mu = \sum_{a\in\actionSet} \gd(a) \cdot \mu_a$.

\begin{definition}\label{def:bisim-parameter-syn}
Let $R \subseteq \stateSet \times \stateSet$ be a symmetric relation. We say that $R$ is \emph{probabilistic $\paramSyn$-bisimulation} if for any $(s,t) \in R$ we have that $\APLabelling(s) = \APLabelling(t)$ and
\begin{align*}
\text{for each} \;\;\;& (\mu_a)_{a\in\actionSet} \;\;\; \text{where $s \tra{a} \mu_a \;$ for each $a\in\actionSet$}\\
\text{there is} \;\;\;& (\nu_a)_{a\in\actionSet} \;\;\; \text{where $t \tra{a} \nu_a \;$ for each $a\in\actionSet$} \\
\text{such that for each} \;\;\;& t \tra{(\nu_a)\;\;} \nu  \\
\text{there is} \;\;\;& s \tra{(\mu_a)\;\;} \mu \;\;\; \text{such that} \;\; \;\; \text{$\mu(\equivclass) =  \nu(\equivclass)$ for each equivalence class $\equivclass \in \partitionset{\stateSet}$},
\end{align*}
Furthermore, we write $s \strongBisim_\paramSyn t$ if there is a probabilistic $\paramSyn$-bisimulation $R$ such that $(s,t) \in R$.
\end{definition}
The fact that this bisimulation preserves $\models_\paramSyn$ can be
proven analogously to Theorem~\ref{thm:preserve-control-syn}.

\begin{theorem}~\label{thm:preserve-paramSyn}
For states $s \strongBisim_\paramSyn t$ and any PCTL formula $\varphi$, we have $s \models_\paramSyn \varphi$ if and only if $t \models_\paramSyn \varphi$.
\end{theorem}

As a final result of this section, we show that 
%we can use either of the bisimulations for both of the cases. Namely that 
these two bisimulations %as well as satisfaction relations 
coincide. 

\begin{theorem}\label{thm:coinsideness}
We have $\strongBisim_\controlSyn \;\; = \;\; \strongBisim_\paramSyn$.% and also $\models_\controlSyn \;\; = \;\; \models_\paramSyn$.
\end{theorem}
%\begin{proof}
%\jk{Needs to be done}
%The first part can be done easily by determinacy of the game. 
%%
%The fact that the bisimulations coincide probably needs a different treatment.
%\end{proof}

Thanks to this result, we denote from now on these coinciding
bisimulations by $\strongBisim_\synthesis$. As a concluding remark,
note that
Definitions~\ref{def:bisim-cooperative},~\ref{def:bisim-control-syn}
and~\ref{def:bisim-parameter-syn} can be seen as the conservative
extension of probabilistic bisimulation for (state-labelled) MDPs. To
see that assume the set of uncertainty for every transition is a
singleton.  Since there is only one choice for the nature, the role of
nature can be safely removed from the definitions. Moreover, it is
worthwhile to note that
Theorems~\ref{thm:coop-bisim},~\ref{thm:preserve-control-syn}
and~\ref{thm:preserve-paramSyn} show the soundness of the
probabilistic bisimulation definitions with respect to
PCTL. Unfortunately, it is shown in~\cite{Segala-thesis, SL94} that
probabilistic bisimulation for probabilistic automata is finer than
PCTL equivalence which leads to the incompleteness in general. Since
MDPs can be seen as a subclass of PAs, it is not difficult to see that
the incompleteness holds also for MDPs.

We also remark that the notions $\strongBisim_\un$ and
$\strongBisim_\synthesis$ are incomparable, as it is for instance
observable in Example~\ref{ex:intro}.  It is shown in the example that
$t\strongBisim_\un\overline{t}$ and
$u\strongBisim_\synthesis\overline{u}$. However it is not hard to
verify that $t\not\!\strongBisim_\synthesis\overline{t}$ and
$u\not\!\strongBisim_\un\overline{u}$.  For the latter, notice that
for example $u$ can evolve to $r$ with probability one by taking
action $b$, whereas $\overline{u}$ cannot simulate.  The former is
noticeable in the situation where the controller wants to maximise the
probability to reach $r$, but the nature declines.  In this case $t$
chooses action $b$ and the nature let it go to $r$ with probability
$0.8$.  Nevertheless the nature can prevent $\overline{t}$ to evolve
into $r$ with probability more than $0.6$, despite the fact which
action has been chosen by $\overline{t}$.

\section{Algorithms}
\label{sec:deciding_strong_bisim}

In this section, we give algorithms for computing bisimulations $\bisim_\un$ and $\bisim_\synthesis$. We show that computing bisimulations in both cases is fixed-parameter tractable.

\begin{example}
Let us start by illustrating the ideas on Example~\ref{ex:intro} from Section~\ref{sec:Introduction}.
\begin{center}
%<<<<<<< .mine
%\input{figures/examples}
%=======
%\input{figures/example4}
%%>>>>>>> .r69
\usetikzlibrary{calc}

\begin{tikzpicture}[x=2.5cm,y=1.2cm,outer sep=0.5mm,
state/.style={draw,circle, inner sep =0.25em,text centered},
trans/.style={font=\footnotesize},
prob/.style={font=\scriptsize,xscale=0.9}
]

\begin{scope}[xshift=14em]

\begin{scope}
\node[state] (top) at (0,0) {$t$};
\node[state] (left) at (-0.3,-2) {$\ell$};
\node[state] (right) at (0.3,-2) {$r$};

\path[->] (top) edge[in=90,out=180,looseness=1] node[above left=-4pt,pos=0.1,trans]{$a$} node[pos=0.4,name=topleft,inner sep=0,outer sep=0]{} node[above=-3pt,sloped,pos=0.8,prob] {$[0.1,0.3]$} (left);

\path[->] (top) edge[in=90,out=0,looseness=1] node[above right=-4pt,pos=0.1,trans]{$b$} node[pos=0.4,name=topright,inner sep=0,outer sep=0]{} node[above=-3pt,pos=0.8,prob,sloped] {$[0,1]$} (right);

\path[->] (topleft) edge[bend right=5] node[above=-3pt,pos=0.2,prob,sloped] {$[0.8,1]$} (right);
\path[->] (topright) edge[bend left=5] node[above=-3pt,pos=0.2,prob,sloped] {$[0.2,0.6]$} (left);
\end{scope}

\begin{scope}[xshift=8em]
\node[state] (top) at (0,0) {$\overline{t}$};
\node[state] (left) at (-0.3,-2) {$\ell$};
\node[state] (right) at (0.3,-2) {$r$};

\path[->] (top) edge[in=90,out=180,looseness=1] node[above left=-4pt,pos=0.1,trans]{$c$} node[pos=0.4,name=topleft,inner sep=0,outer sep=0]{} node[above=-3pt,sloped,pos=0.8,prob] {$[0.1,1]$} (left);

\path[->] (top) edge[in=90,out=0,looseness=1] node[above right=-4pt,pos=0.1,trans]{$d$} node[pos=0.4,name=topright,inner sep=0,outer sep=0]{} node[above=-3pt,pos=0.8,prob,sloped] {$[0,0.8]$} (right);

\path[->] (topleft) edge[bend right=5] node[above=-3pt,pos=0.2,prob,sloped] {$[0.4,0.9]$} (right);
\path[->] (topright) edge[bend left=5] node[above=-3pt,pos=0.2,prob,sloped] {$[0.2,0.4]$} (left);
\end{scope}

\draw [dotted] (2.15,0.3) -- (2.15,-2.3);

\end{scope}

\begin{scope}[xshift=34em]

\begin{scope}
\node[state] (top) at (0,0) {$u$};
\node[state] (left) at (-0.3,-2) {$\ell$};
\node[state] (right) at (0.3,-2) {$r$};

\path[->] (top) edge[in=90,out=180,looseness=1] node[above left=-4pt,pos=0.1,trans]{$a$} node[pos=0.4,name=topleft,inner sep=0,outer sep=0]{} node[above=-3pt,sloped,pos=0.8,prob] {$[0.1,0.6]$} (left);

\path[->] (top) edge[in=90,out=0,looseness=1] node[above right=-4pt,pos=0.1,trans]{$b$} node[pos=0.4,name=topright,inner sep=0,outer sep=0]{} node[above=-3pt,pos=0.8,prob,sloped] {$[0,1]$} (right);

\path[->] (topleft) edge[bend right=5] node[above=-3pt,pos=0.2,prob,sloped] {$[0,1]$} (right);
\path[->] (topright) edge[bend left=5] node[above=-3pt,pos=0.2,prob,sloped] {$[0,0.6]$} (left);
\end{scope}

\begin{scope}[xshift=8em]
\node[state] (top) at (0,0) {$\overline{u}$};
\node[state] (left) at (-0.3,-2) {$\ell$};
\node[state] (right) at (0.3,-2) {$r$};

\path[->] (top) edge[in=90,out=180,looseness=1] node[above left=-4pt,pos=0.1,trans]{$a$} node[pos=0.4,name=topleft,inner sep=0,outer sep=0]{} node[above=-3pt,sloped,pos=0.8,prob] {$[0.1,0.6]$} (left);

\path[->] (top) edge[in=90,out=0,looseness=1] node[above right=-4pt,pos=0.1,trans]{$c$} node[pos=0.4,name=topright,inner sep=0,outer sep=0]{} node[above=-3pt,pos=0.8,prob,sloped] {$[0,1]$} (right);

\path[->] (topleft) edge[bend right=5] node[above=-3pt,pos=0.2,prob,sloped] {$[0,1]$} (right);
\path[->] (topright) edge[bend left=5] node[above=-3pt,pos=0.2,prob,sloped] {$[0.1,0.8]$} (left);
\end{scope}

% \draw [dotted] (1.55,0.3) -- (1.55,-2.3);

\end{scope}

\end{tikzpicture}

\begin{tikzpicture}[x=2.5cm,y=1.2cm,outer sep=0.5mm,
state/.style={draw,circle, inner sep =0.25em,text centered},
trans/.style={font=\footnotesize},
prob/.style={font=\scriptsize,xscale=0.8,yscale=0.9},
eq/.style={prob,draw,rectangle,inner sep=1pt,yshift=-2pt}
]

\begin{scope}[xshift=14em]

\path[draw,->,ultra thick] (0.55,1.1) -- node[right] {$t \bisim_\un \overline{t}$} (0.55,0.3);

\begin{scope}
\node[state] (top) at (0,0) {$t$};
\node[] (left) at (-0.3,-0.5) {};
\node[] (right) at (0.3,-0.5) {};

\path[->] (top) edge[in=90,out=180,looseness=1] node[above left=-4pt,pos=0.3,trans]{$a$} (left);
\path[->] (top) edge[in=90,out=0,looseness=1] node[above right=-4pt,pos=0.3,trans]{$b$} (right);

\node[eq] at ($(left)+(0,-0.3)$) {$\begin{aligned}
     \ell+r &= 1\\
     0.1 \leq \ell &\leq 0.3 \\
     0.8 \leq r &\leq 1 \\
  \end{aligned}$};

\node[eq] at ($(right)+(-0.06,-0.3)$) {$\begin{aligned}
     \ell+r &= 1\\
     0.2 \leq \ell &\leq 0.6 \\
     0 \leq r &\leq 1 \\
  \end{aligned}$};

% getting the scale right again
\begin{scope}[xshift=-3em,yshift=-7.5em,xscale=0.32,yscale=0.666]
\node [trans] at (0.9,1.3) {$P_1$};

\draw[->] (0,0) -- node[below=-2pt,pos=1.1,prob]{$\ell$} (1.2,0);
\draw[->] (0,0) -- node[left=-2pt,pos=1.1,prob]{$r$} (0,1.2);
\draw[-] (1,-0.04) -- node[below,pos=0.85,prob,scale=0.8]{$1$} (1,0.04);
\draw[-] (-0.04,1) -- node[left,pos=0.85,prob,scale=0.8]{$1$} (0.04,1);

\draw[-,densely dotted, thin] (1,0) -- (0,1);

\foreach \x in {.1,.3} {
	\draw[-] (\x,-0.04) -- node[below,pos=0.85,prob,scale=0.8]{$\x$} (\x,0.04);
	\draw[-,densely dotted, thin] (\x,0) -- (\x,1-\x);
};
\foreach \y in {.8} {
	\draw[-] (-0.04,\y) -- node[left,pos=0.85,prob,scale=0.8]{$\y$} (0.04,\y);
	\draw[-,densely dotted, thin] (0,\y) -- (1-\y,\y);
};

\foreach \a/\b/\c in {0.1/0.2/0.035} {
	\draw [-,thick] (\a,1-\a) -- node [above=0pt,sloped,rotate=-17.5,prob,xscale=0.9] {$[\a,\b]$} (\b,1-\b);
	\draw [-,thick] (\a-\c,1-\a-\c) -- (\a+\c,1-\a+\c);
	\draw [-,thick] (\b-\c,1-\b-\c) -- (\b+\c,1-\b+\c);
};
\end{scope}
% getting the scale right again
\begin{scope}[xshift=0.5em,yshift=-7.5em,xscale=0.32,yscale=0.666]
\node [trans] at (0.9,1.3) {$P_2$};

\draw[->] (0,0) -- node[below=-2pt,pos=1.1,prob]{$\ell$} (1.2,0);
\draw[->] (0,0) -- node[left=-2pt,pos=1.1,prob]{$r$} (0,1.2);
\draw[-] (1,-0.04) -- node[below,pos=0.85,prob,scale=0.8]{$1$} (1,0.04);
\draw[-] (-0.04,1) -- node[left,pos=0.85,prob,scale=0.8]{$1$} (0.04,1);

\draw[-,densely dotted, thin] (1,0) -- (0,1);

\foreach \x in {.2,.6} {
	\draw[-] (\x,-0.04) -- node[below,pos=0.85,prob,scale=0.8]{$\x$} (\x,0.04);
	\draw[-,densely dotted, thin] (\x,0) -- (\x,1-\x);
};
\foreach \y in {} {
	\draw[-] (-0.04,\y) -- node[left,pos=0.85,prob,scale=0.8]{$\y$} (0.04,\y);
	\draw[-,densely dotted, thin] (0,\y) -- (1-\y,\y);
};

\foreach \a/\b/\c in {0.2/0.6/0.035} {
	\draw [-,thick] (\a,1-\a) -- node [above=0pt,sloped,rotate=-17.5,prob,xscale=0.9] {$[\a,\b]$} (\b,1-\b);
	\draw [-,thick] (\a-\c,1-\a-\c) -- (\a+\c,1-\a+\c);
	\draw [-,thick] (\b-\c,1-\b-\c) -- (\b+\c,1-\b+\c);
};
\end{scope}

\end{scope}

\begin{scope}[xshift=8em]
\node[state] (top) at (0,0) {$\overline{t}$};
\node[] (left) at (-0.3,-0.5) {};
\node[] (right) at (0.3,-0.5) {};

\path[->] (top) edge[in=90,out=180,looseness=1] node[above left=-4pt,pos=0.3,trans]{$c$} (left);
\path[->] (top) edge[in=90,out=0,looseness=1] node[above right=-4pt,pos=0.3,trans]{$d$} (right);

\node[eq] at ($(left)+(0.02,-0.3)$) {$\begin{aligned}
     \ell+r &= 1\\
     0.1 \leq \ell &\leq 1 \\
     0.4 \leq r &\leq 0.9 \\
  \end{aligned}$};

\node[eq] at ($(right)+(-0.03,-0.3)$) {$\begin{aligned}
     \ell+r &= 1\\
     0.2 \leq \ell &\leq 0.4 \\
     0 \leq r &\leq 0.8 \\
  \end{aligned}$};

% getting the scale right again
\begin{scope}[xshift=-3em,yshift=-7.5em,xscale=0.32,yscale=0.666]
\node [trans] at (0.9,1.3) {$P_3$};

\draw[->] (0,0) -- node[below=-2pt,pos=1.1,prob]{$\ell$} (1.2,0);
\draw[->] (0,0) -- node[left=-2pt,pos=1.1,prob]{$r$} (0,1.2);
\draw[-] (1,-0.04) -- node[below,pos=0.85,prob,scale=0.8]{$1$} (1,0.04);

\draw[-,densely dotted, thin] (1,0) -- (0,1);

\foreach \x in {.1} {
	\draw[-] (\x,-0.04) -- node[below,pos=0.85,prob,scale=0.8]{$\x$} (\x,0.04);
	\draw[-,densely dotted, thin] (\x,0) -- (\x,1-\x);
};
\foreach \y in {.4,.9} {
	\draw[-] (-0.04,\y) -- node[left,pos=0.85,prob,scale=0.8]{$\y$} (0.04,\y);
	\draw[-,densely dotted, thin] (0,\y) -- (1-\y,\y);
};

\foreach \a/\b/\c in {0.1/0.6/0.035} {
	\draw [-,thick] (\a,1-\a) -- node [above=0pt,sloped,rotate=-17.5,prob,xscale=0.9] {$[\a,\b]$} (\b,1-\b);
	\draw [-,thick] (\a-\c,1-\a-\c) -- (\a+\c,1-\a+\c);
	\draw [-,thick] (\b-\c,1-\b-\c) -- (\b+\c,1-\b+\c);
};
\end{scope}
% getting the scale right again
\begin{scope}[xshift=0.7em,yshift=-7.5em,xscale=0.32,yscale=0.666]
\node [trans] at (0.9,1.3) {$P_4$};

\draw[->] (0,0) -- node[below=-2pt,pos=1.1,prob]{$\ell$} (1.2,0);
\draw[->] (0,0) -- node[left=-2pt,pos=1.1,prob]{$r$} (0,1.2);
\draw[-] (1,-0.04) -- node[below,pos=0.85,prob,scale=0.8]{$1$} (1,0.04);
\draw[-] (-0.04,1) -- node[left,pos=0.85,prob,scale=0.8]{$1$} (0.04,1);

\draw[-,densely dotted, thin] (1,0) -- (0,1);

\foreach \x in {.2,.4} {
	\draw[-] (\x,-0.04) -- node[below,pos=0.85,prob,scale=0.8]{$\x$} (\x,0.04);
	\draw[-,densely dotted, thin] (\x,0) -- (\x,1-\x);
};
\foreach \y in {.8} {
	\draw[-] (-0.04,\y) -- node[left,pos=0.85,prob,scale=0.8]{$\y$} (0.04,\y);
	\draw[-,densely dotted, thin] (0,\y) -- (1-\y,\y);
};

\foreach \a/\b/\c in {0.2/0.4/0.035} {
	\draw [-,thick] (\a,1-\a) -- node [above=0pt,sloped,rotate=-17.5,prob,xscale=0.9] {$[\a,\b]$} (\b,1-\b);
	\draw [-,thick] (\a-\c,1-\a-\c) -- (\a+\c,1-\a+\c);
	\draw [-,thick] (\b-\c,1-\b-\c) -- (\b+\c,1-\b+\c);
};
\end{scope}

\end{scope}

\draw [dotted] (2.155,1.1) -- (2.155,-2.3);

\end{scope}

\begin{scope}[xshift=34em]

\path[draw,->,ultra thick] (0.45,1.1) -- node[right] {$u \bisim_\synthesis \overline{u}$} (0.45,0.3);

\begin{scope}
\node[state] (top) at (0,0) {$u$};
\node[] (left) at (-0.3,-0.5) {};
\node[] (right) at (0.3,-0.5) {};

\path[->] (top) edge[in=90,out=180,looseness=1] node[above left=-4pt,pos=0.3,trans]{$a$} (left);
\path[->] (top) edge[in=90,out=0,looseness=1] node[above right=-4pt,pos=0.3,trans]{$b$} (right);

\node[eq] at ($(left)+(0.02,-0.3)$) {$\begin{aligned}
     \ell+r &= 1\\
     0.1 \leq \ell &\leq 0.6 \\
     0 \leq r &\leq 1 \\
  \end{aligned}$};

\node[eq] at ($(right)+(-0.08,-0.3)$) {$\begin{aligned}
     \ell+r &= 1\\
     0 \leq \ell &\leq 0.6 \\
     0 \leq r &\leq 1 \\
  \end{aligned}$};

% getting the scale right again
\begin{scope}[xshift=-3em,yshift=-7.5em,xscale=0.32,yscale=0.666]
\node [trans] at (0.9,1.3) {$P_1$};

% \draw[->] (0,0) -- node[below=-2pt,pos=1.1,prob]{$\ell$} (1.2,0);
% \draw[->] (0,0) -- node[left=-2pt,pos=1.1,prob]{$r$} (0,1.2);

\draw[-,densely dotted, thin] (1,0) -- (0,1);

\foreach \a/\b/\c in {0.1/0.6/0.035} {
	\draw [-,thick] (\a,1-\a) -- node [above=0pt,sloped,rotate=-17.5,prob,xscale=0.9] {$[\a,\b]$} (\b,1-\b);
	\draw [-,thick] (\a-\c,1-\a-\c) -- (\a+\c,1-\a+\c);
	\draw [-,thick] (\b-\c,1-\b-\c) -- (\b+\c,1-\b+\c);
};
\end{scope}
% getting the scale right again
\begin{scope}[xshift=0.3em,yshift=-7.5em,xscale=0.32,yscale=0.666]
\node [trans] at (0.9,1.3) {$P_2$};

% \draw[->] (0,0) -- node[below=-2pt,pos=1.1,prob]{$\ell$} (1.2,0);
% \draw[->] (0,0) -- node[left=-2pt,pos=1.1,prob]{$r$} (0,1.2);

\draw[-,densely dotted, thin] (1,0) -- (0,1);

\foreach \a/\b/\c in {0/0.6/0.035} {
	\draw [-,thick] (\a,1-\a) -- node [above=0pt,sloped,rotate=-17.5,prob,xscale=0.9] {$[\a,\b]$} (\b,1-\b);
	\draw [-,thick] (\a-\c,1-\a-\c) -- (\a+\c,1-\a+\c);
	\draw [-,thick] (\b-\c,1-\b-\c) -- (\b+\c,1-\b+\c);
};
\end{scope}

\end{scope}

\begin{scope}[xshift=8em]
\node[state] (top) at (0,0) {$\overline{u}$};
\node[] (left) at (-0.3,-0.5) {};
\node[] (right) at (0.3,-0.5) {};

\path[->] (top) edge[in=90,out=180,looseness=1] node[above left=-4pt,pos=0.3,trans]{$a$} (left);
\path[->] (top) edge[in=90,out=0,looseness=1] node[above right=-4pt,pos=0.3,trans]{$c$} (right);

\node[eq] at ($(left)+(0.02,-0.3)$) {$\begin{aligned}
     \ell+r &= 1\\
     0.1 \leq \ell &\leq 0.6 \\
     0 \leq r &\leq 1 \\
  \end{aligned}$};

\node[eq] at ($(right)+(-0.05,-0.3)$) {$\begin{aligned}
     \ell+r &= 1\\
     0.1 \leq \ell &\leq 0.8 \\
     0 \leq r &\leq 1 \\
  \end{aligned}$};

% getting the scale right again
\begin{scope}[xshift=-3em,yshift=-7.5em,xscale=0.32,yscale=0.666]
\node [trans] at (0.9,1.3) {$P_1$};

% \draw[->] (0,0) -- node[below=-2pt,pos=1.1,prob]{$\ell$} (1.2,0);
% \draw[->] (0,0) -- node[left=-2pt,pos=1.1,prob]{$r$} (0,1.2);

\draw[-,densely dotted, thin] (1,0) -- (0,1);

\foreach \a/\b/\c in {0.1/0.6/0.035} {
	\draw [-,thick] (\a,1-\a) -- node [above=0pt,sloped,rotate=-17.5,prob,xscale=0.9] {$[\a,\b]$} (\b,1-\b);
	\draw [-,thick] (\a-\c,1-\a-\c) -- (\a+\c,1-\a+\c);
	\draw [-,thick] (\b-\c,1-\b-\c) -- (\b+\c,1-\b+\c);
};
\end{scope}
% getting the scale right again
\begin{scope}[xshift=0.3em,yshift=-7.5em,xscale=0.32,yscale=0.666]
\node [trans] at (0.9,1.3) {$P_3$};

% \draw[->] (0,0) -- node[below=-2pt,pos=1.1,prob]{$\ell$} (1.2,0);
% \draw[->] (0,0) -- node[left=-2pt,pos=1.1,prob]{$r$} (0,1.2);

\draw[-,densely dotted, thin] (1,0) -- (0,1);

\foreach \a/\b/\c in {0.1/0.8/0.035} {
	\draw [-,thick] (\a,1-\a) -- node [above=0pt,sloped,rotate=-17.5,prob,xscale=0.9] {$[\a,\b]$} (\b,1-\b);
	\draw [-,thick] (\a-\c,1-\a-\c) -- (\a+\c,1-\a+\c);
	\draw [-,thick] (\b-\c,1-\b-\c) -- (\b+\c,1-\b+\c);
};
\end{scope}

\end{scope}

\end{scope}

\end{tikzpicture}
\end{center} 
%
%The naive approach is depicted on the left -- namely to treat the set of intervals as the set of transition labels and apply in this symbolic setting the standard non-probabilistic strong bisimulation (where both state labels and transition labels are checked). Interestingly, we prove in Section~\ref{sec:alg-symbolic} that under specific assumptions, this approach actually works. Note it is also the case for $s$ and $\bar{s}$ which are clearly by this method not bisimilar.
%
%In the general setting, w
The general sketch of the algorithm is as follows. We need to construct the polytopes of probability distributions offered by the actions; in our examples the polytopes are just line segments in two-dimensional space. We get $t \bisim_\un \overline{t}$ since the \emph{convex hull} of $P_1$ and $P_2$ equals to the \emph{convex hull} of $P_3$ and $P_4$. Similarly, we get $u \bisim_\synthesis \overline{u}$ since $u$ and $\overline{u}$ have the same set of \emph{minimal} polytopes w.r.t. set inclusion.

\end{example}

%\subsection{Algorithms for general models}
%
%In this subsection we address the algorithms for the general case and prove the following theorem.

Let us state the results formally. Let us fix $\imdp = (\stateSet, \actionSet, \APSet,\APLabelling, \intTransitionProbability)$ where $b$ is the maximal number of different actions $\max_{s\in\stateSet} | \{ \intTransitionProbability(s,a,\cdot) \mid a \in\actionSet \} |$, $f$ is the maximal support of an action $\max_{s\in\stateSet, a\in\actionSet} | \{s' \mid \intTransitionProbability(s,a,s') \neq [0,0] \} |$, and $|\imdp|$ denotes the size of the representation using adjacency lists for non-zero elements of $\intTransitionProbability$ where we assume that the interval bounds are rational numbers encoded symbolically in binary.

\begin{theorem}
There is an algorithm that computes $\bisim_\un$ in time polynomial in $|\imdp|$ and exponential in $f$.
There is also an algorithm that computes $\bisim_\synthesis$ in time polynomial in $|\imdp|$ and exponential in $f$ and $b$.
\end{theorem}

Computing both bisimulations follows the standard partition refinement approach~\cite{PT87,KS90,CS02}, formalized by the procedure $\ComputeBisimProc$ in Algorithm~\ref{fig:algoForBisim}. Namely, we start with $\partitioningord$ being the complete relation and iteratively remove from $\partitioningord$ pairs of states that violate the definition of bisimulation with respect to $\partitioningord$.
The core part is finding out whether two states ``violate the definition of bisimulation''. This is where the algorithms for the two bisimulations differ.

%
%---
%
%In order to compute the bisimulation $\strongBisim_\controlSyn$ for the \IMDP{} $\imdp$ we follow the standard partition refinement approach~\cite{CS02,PLS00,KS90,PT87,HT12,EHKTZ13} depicted in Figure~\ref{fig:algoForBisim}: In the procedure $\ComputeBisimProc$ the bisimulation $\strongBisim_\controlSyn$ and the IMDP{} $\imdp$ are taken as inputs. Then it starts with  the partitionin $\partitioningord = \setcond{(s,t) \in \stateSet \times \stateSet}{\functioneval{\APLabelling}{s_{1}} = \functioneval{\APLabelling}{s_{2}}}$ and iteratively builds the set $\quotienting{\stateSet}{\mathord{\strongBisim_\controlSyn}}$ and refines until $\partitioningord$ satisfies the definition of $\strongBisim_\controlSyn$. When the algorithm reaches to the refinement phase, it checks if for each partition all pairs of states are bisimilar. If this condition fails for a pair $s, t$ then the current partition is split into two partitions: one including $s$ and all the other states for which this condition is held with respect to the transitions available for $s$ and the other including $t$ and the remaining states. At termination the algorithm reaches to a fixed point, i.e., it computes the bisimilarity $\strongBisim_\controlSyn$. 

\begin{algorithm}[t]
	\small
	\centering
	\begin{tikzpicture}[baseline]
		\node[draw, text height=2ex, minimum width=0.48\linewidth, inner sep = 0pt, minimum height=3ex] (heading) {\progHeader{$\ComputeBisimProc(\imdp)$}};
		\node[below=0.5ex of heading] (text) {%
			\begin{minipage}[h]{0.48\linewidth}
				\begin{algorithmic}[1]
				\STATE{$\partitioningord \gets \setcond{(s,t) \in \stateSet \times \stateSet}{\functioneval{\APLabelling}{s} = \functioneval{\APLabelling}{t}}$;}
				\REPEAT
					\STATE{$\partitioningord' \gets  \partitioningord$;}
%					\STATE{$(\classToSplit,\splitAction,\splitDistribution) \gets \FindSplitProc(\imdp, \partitioningord)$;}
					\FORALL{$s \in \stateSet$}
						\STATE{$D \gets \emptyset$};
						\FORALL{$t \in \relclass{s}{\partitioningord}$}
							\IF{$ \ViolateProc(s,t,\partitioningord)$} 
								\STATE{$D \gets D \cup \{t\}$;}
							\ENDIF
						\ENDFOR
						\STATE{split $\relclass{s}{\partitioningord}$ in $\partitioningord$ into $D$ and $\relclass{s}{\partitioningord}\setminus D$};
					\ENDFOR
				\UNTIL{$\partitioningord = \partitioningord'$};
				\RETURN{$\partitioningord$};
				\end{algorithmic}
			\end{minipage}
		}; 
%		\coordinate (END) at ($(heading.north west)!24ex!(heading.south	west)$); 
%		\draw (heading.south west) to (END) {}; 
%		\draw (END) to ($(END)+(2em,0)$) {};
	\end{tikzpicture}
\begin{minipage}[t]{0.5\linewidth}
\begin{tikzpicture}[baseline]
		\node[draw, text height=2ex, minimum width=\linewidth, inner sep = 0pt, minimum height=3ex] (heading) {\progHeader{$\ViolateProc_\un(s,t,\partitioningord)$}}; 
		\node[below=0.5ex of heading] (text) {%
			\begin{minipage}[h]{\linewidth}
				\begin{algorithmic}[1]
%					\FORALL{$a\in \actionSet$}
%						\IF{$\paruncertainty{s}{a}{\partitioningord} \not\subseteq \parcombuncertainty{t}{\partitioningord}$ \OR
%						$\paruncertainty{t}{a}{\partitioningord} \not\subseteq \parcombuncertainty{s}{\partitioningord}$} 
%							\RETURN{true};
%						\ENDIF
%					\ENDFOR
%					\IF{$\parcombuncertainty{s}{\partitioningord} \neq  \parcombuncertainty{t}{\partitioningord}$} 
%							\RETURN{true};
%					\ELSE
%					\RETURN{false};
%					\ENDIF
					\RETURN{$\parcombuncertainty{s}{\partitioningord} \neq  \parcombuncertainty{t}{\partitioningord}$};
				\end{algorithmic}
			\end{minipage}
		}; 
%		\coordinate (END) at ($(heading.north west)!20ex!(heading.south west)$); 
%		\draw (heading.south west) to (END) {}; 
%		\draw (END) to ($(END)+(3em,0)$) {};
    \end{tikzpicture}
	\begin{tikzpicture}
		\node[draw, text height=2ex, minimum width=\linewidth, inner sep = 0pt, minimum height=3ex] (heading) {\progHeader{$\ViolateProc_\synthesis(s,t,\partitioningord)$}}; 
		\node[below=0.5ex of heading] (text) {%
			\begin{minipage}[h]{\linewidth}
				\begin{algorithmic}[1]
					\STATE{$S,T \gets \emptyset$;}
					\FORALL{$a\in \actionSet$}
						\IF{$\MinimalProc(s,a,\partitioningord)$ \hfill \texttt{// $\paruncertainty{s}{a}{\partitioningord}$ strictly m.?}}
							\STATE{$S \gets S \cup \{\paruncertainty{s}{a}{\partitioningord}\}$;}
						\ENDIF
						\IF{$\MinimalProc(t,a,\partitioningord)$}
							\STATE{$T \gets T \cup \{\paruncertainty{t}{a}{\partitioningord}\}$;}
						\ENDIF
					\ENDFOR
					\RETURN{$S \neq T$};
				\end{algorithmic}
			\end{minipage}
		}; 
%		\coordinate (END) at ($(heading.north west)!20ex!(heading.south west)$); 
%		\draw (heading.south west) to (END) {}; 
%		\draw (END) to ($(END)+(3em,0)$) {};
    \end{tikzpicture}
    \end{minipage}
    \begin{tikzpicture}
        		\node[draw, text height=2ex, minimum width=\linewidth, inner sep = 0pt, minimum height=3ex] (heading) {\progHeader{$\MinimalProc(s,a,\partitioningord)$}}; 
        		\node[below=0.5ex of heading] (text) {%
        			\begin{minipage}[h]{\linewidth}
        				\begin{algorithmic}[1]
        					\STATE $k \gets |\actionSet| - 1$;
        					\STATE{$C_1,\ldots,C_k \gets $ compute the sets of corners of other polytopes;}
    %    					\STATE{$\rho_1,\ldots,\rho_k \gets $ fresh variables for $\rho$ in $\paruncertainty{s}{\rho}{\partitioningord}$;}
    						\STATE{$B \gets (\vec{1},\vec{-1})$;	\hfill \texttt{// constraints on $\rho$ such that $B\rho + d \geq 0$ implies $\paruncertainty{s}{\rho}{\partitioningord} \subseteq \paruncertainty{s}{a}{\partitioningord}$}}
    						\STATE{$d \gets (-1,1)$;  \hfill \texttt{// initially, $B\rho + d \geq 0$ implies $\sum \rho = 1$, i.e. $\rho \in \distrib{\{1,\ldots,k\}}$}}
        					\FORALL{$c_1,\ldots,c_k \in C_1\times\cdots\times C_k$}
        						\STATE{$R \gets \emptyset$;}
        						\FORALL{intersections $\vec{x}$ of $\paruncertainty{s}{a}{\partitioningord}$ with the line segment from $v_i$ to $v_j$ for some $i \neq j$}
	           						\STATE{$R \gets R \cup \{(r_1,\cdots,r_k)\}$ where $r_i\cdot c_i + r_j\cdot c_j = \vec{x}$ and $r_\ell = 0$ for $\ell \not\in \{i,j\}$;}
        						\ENDFOR
        						\FORALL{facets $F$ of the convex hull of $R$}
        							\STATE{add to matrices $B,d$ a constraint corresponding to the half-space given by $F$ that includes $R$;}
        						\ENDFOR
        					\ENDFOR
        					\RETURN{($B\rho + d = 0$ not feasible); \hfill \texttt{// no intersection of the convex hulls of all sets $R$?}}
        				\end{algorithmic}
        			\end{minipage}
        		}; 
        %		\coordinate (END) at ($(heading.north west)!20ex!(heading.south west)$); 
        %		\draw (heading.south west) to (END) {}; 
        %		\draw (END) to ($(END)+(3em,0)$) {};
            \end{tikzpicture}
	\caption{Probabilistic bisimulation algorithm for interval \MDP{}s}
	\label{fig:algoForBisim}
\end{algorithm}

\subsection{Cooperative resolution of non-determinism of the bisimulation $\bisim_\un$}
Let us first address $\bisim_\un$ where the violation is checked by the procedure $\ViolateProc_\un$.
We show that this amounts to checking inclusion of polytopes defined as follows. 
Recall that for $s \in\stateSet$ and an action $a\in\actionSet$, $\uncertainty{s}{a}$ denotes the polytope of feasible successor distributions over \emph{states} with respect to taking the action $a$ in the state $s$.
By $\paruncertainty{s}{a}{\partitioningord}$, we denote the polytope of feasible successor distributions over \emph{equivalence classes} of $\partitioningord$ with respect to taking the action $a$ in the state $s$. Formally, for $\mu\in\distrib{\partitionset[\partitioningord]{\stateSet}}$ we set $\mu \in \paruncertainty{s}{a}{\partitioningord}$ if we have 
\begin{equation*}\mu(\eqclass) \in \Big [
{\sum_{s'\in \eqclass}{\inf \; \intTransitionProbability(s, a, s')}}
{\sum_{s'\in \eqclass}{\sup \; \intTransitionProbability(s, a, s')}} \Big ]
\qquad \text{for each $\eqclass \in \partitionset[\partitioningord]{\stateSet}.$}
\end{equation*}
Note that we require that the probability of each class $\eqclass$ must be in the interval of the sum of probabilities that can be assigned to states of $\eqclass$. 
Furthermore, we define $\parcombuncertainty{s}{\partitioningord}$ as the convex hull of $\bigcup_{a\in\actionSet} \paruncertainty{s}{a}{\partitioningord}$. It is the set of feasible successor distributions over $\partitionset[\partitioningord]{\stateSet}$ with respect to taking an \emph{arbitrary} distribution over actions in state $s$. 
As specified in the procedure $\ViolateProc_\un$, we show that it suffices to check equality of these polytopes.
%Let us show the correctness of the algorithm with the procedure .

%It is easy to see that $\paruncertainty{s}{a}{\partitioningord}$ is obtained by linearly mapping $\uncertainty{s}{a}$ into a $|\partitionset[\partitioningord]{\stateSet}|$-dimensional space that only sum up the corresponding dimensions.

\begin{proposition}
We have $s\bisim_\un t$ if and only if $\APLabelling(s) = \APLabelling(t)$ and $\parcombuncertainty{s}{\bisim_\un} =  \parcombuncertainty{t}{\bisim_\un}$.
\end{proposition}
\begin{proof}
Let us first introduce one notation. For each distribution $\mu \in \distrib{\stateSet}$, let $\overline{\mu}\in \distrib{\partitionset[\bisim_\un]{\stateSet}}$ denote the corresponding distribution such that $\overline{\mu}(\eqclass) = \sum_{s\in\eqclass} \mu(s)$.
As regards the ``if'' part, for each choice $s \tra{} \mu$, we have $\overline{\mu} \in \parcombuncertainty{s}{\bisim_\un}$. Similarly, for each $\gd \in \parcombuncertainty{t}{\bisim_\un}$, there is a choice $t \tra{} \nu$ such that $\overline{\nu} = \gd$. Hence, $s\bisim_\un t$. As regards the ``only if'' part, let us assume that there is a distribution $\gd$ over equivalence classes such that, say $\gd \in \parcombuncertainty{s}{\bisim_\un} \setminus \parcombuncertainty{t}{\bisim_\un}$. There must be a choice $s \tra{} \mu$ such that $\overline{\mu} = \gd$ and there is no choice $t \tra{} \nu$ such that $\overline{\nu} = \gd$. Hence, $s\not\bisim_\un t$.
\end{proof}

\paragraph{Complexity}
Given an \IMDP{} $\imdp$, let $\setcardinality{S}=n$, $\setcardinality{A}=m$, $b$ be the maximal number of different actions $\max_{s\in\stateSet} | \{ \intTransitionProbability(s,a,\cdot) \mid a \in\actionSet \} |$, and $f$ be the maximal support of an action $\max_{s\in\stateSet, a\in\actionSet} | \{s' \mid \intTransitionProbability(s,a,s') \neq [0,0] \} |$.

It is easy to see that the procedure $\ViolateProc_\un$ is called at most $n^3$-times. Each polytope $\paruncertainty{s}{a}{\partitioningord}$ has at most $C = f\cdot 2^{f-1}$ corners, computing the convex hull $\parcombuncertainty{s}{\partitioningord}$ takes $\bigO{ (bC)^2 }$ time~\cite{Chand:1970:ACP}. Checking inclusion of two polytopes then can be done in time polynomial~\cite{Subramani09} in the number of corners of these two polytopes. %\jk{Hassan, could you possible fill in the reference, please? How is this inclusion checked when we have the corners explicitely? Can we specify the complexity in more detail?}
%
%In total, the algorithm is polynomial in $|\imdp|$ and exponential in $f$.
%
In total, computing of $\bisim_\un$ can be done in time $|\imdp|^{\bigO{1}} \cdot 2^{\bigO{f}}$.

\subsection{Competitive resolution of non-determinism of the bisimulation $\bisim_\synthesis$} In this case, the violation of bisimilarity of $s$ and $t$ with respect to $\partitioningord$ is addressed by the procedure $\ViolateProc_\synthesis$. Here, we check that $s$ and $t$ have the same set of strictly minimal polytopes. For a state $s$, an action $a\in\actionSet$, and an equivalence $\partitioningord \subseteq \stateSet\times\stateSet$, we say that $\paruncertainty{s}{a}{\partitioningord}$ is \emph{strictly minimal} if no convex combination of the remaining polytopes of $s$ is a subset of $\paruncertainty{s}{a}{\partitioningord}$. More precisely, if for no distribution $\gd \in \distrib{\actionSet \setminus \{a\}}$, we have  $\paruncertainty{s}{\rho}{\partitioningord} \subseteq \paruncertainty{s}{a}{\partitioningord}$ where $\paruncertainty{s}{\rho}{\partitioningord}$ denotes the polytope $\{\sum_{b\in\actionSet\setminus \{a\}} \gd(b) \cdot \vec{x}_b \mid \text{each} \; \vec{x}_b \in \paruncertainty{s}{b}{\partitioningord} \}$.

\begin{proposition}
We have $s\bisim_\synthesis t$ if and only if $\APLabelling(s) = \APLabelling(t)$ and $\{\paruncertainty{s}{a}{\bisim_\synthesis} \mid a\in\actionSet,\paruncertainty{s}{a}{\bisim_\synthesis} \text{is strictly minimal}\} = \{\paruncertainty{t}{a}{\bisim_\synthesis} \mid a\in\actionSet,\paruncertainty{t}{a}{\bisim_\synthesis} \text{is strictly minimal}\}$.
%
%the set of strictly minimal polytopes for every $a\in\actionSet$ we have $\paruncertainty{s}{a}{\bisim_\un} \subseteq \parcombuncertainty{t}{\bisim_\un}$ and $\paruncertainty{t}{a}{\bisim_\un} \subseteq \parcombuncertainty{s}{\bisim_\un}$.
\end{proposition}
\begin{proof}
We first address the ``if'' part. 
 For each choice of {\Envir} $ (\mu_a)_{a\in\actionSet}$ where each $s \tra{a} \mu_a \;$, let $M = \{\overline{\mu}_a \mid a\in\actionSet\}$ and $M' \subseteq M$ be the subset where each distribution lies within some strictly minimal polytope $\paruncertainty{s}{b}{\bisim_\synthesis}$. Because the strictly minimal polytopes coincide, we can construct a choice of {\Envir} $N = (\nu_a)_{a\in\actionSet}$ such that $N = \{\overline{\nu}_a \mid a\in\actionSet\} = M'$. Because $N \subseteq M$, it is easy to see that for each $t \tra{(\nu_a)\;\;} \nu$ there is $s \tra{(\mu_a)\;\;} \mu$ such that $\mu(\equivclass) =  \nu(\equivclass)$ for each $\equivclass \in \partitionset{\stateSet}$.
 
 As regards the ``only if'' part, let us assume that there is, say in $t$, a strictly minimal polytope $\paruncertainty{t}{b}{\bisim_\synthesis}$ that is not in the set of strictly minimal polytopes for $s$. There is a choice of {\Envir} $(\mu_a)_{a\in\actionSet}$ for state $s$ such that no convex combination of elements of $M = \{\overline{\mu}_a \mid a\in\actionSet\}$ lies in $\paruncertainty{t}{b}{\bisim_\synthesis}$; in particular no element of $M$ lies in $\paruncertainty{t}{b}{\bisim_\synthesis}$. For any choice of {\Envir} $(\nu_a)_{a\in\actionSet}$ for state $t$, 
%  then results in a set $N = \{\bar{\nu}_a \mid a\in\actionSet\}$ such that $N \not\subseteq M$ because 
 $\overline{\nu}_b$ is not a convex combination of elements from $M$. Thus, if {\Sched} chooses action $b$, there is no $s \tra{(\mu_a)\;\;} \mu$ such that $\mu(\equivclass) =  \nu_b(\equivclass)$ for each $\equivclass \in \partitionset{\stateSet}$ and it does \emph{not} hold $s \bisim_\synthesis t$.
\end{proof}

Next, we need to address how to compute whether a polytope is strictly minimal. We construct $B$ and $d$ such that $B\rho + d \geq 0$ implies $\paruncertainty{s}{\rho}{\partitioningord} \subseteq \paruncertainty{s}{a}{\partitioningord}$. Checking of strictly minimality then reduces to checking feasibility of this linear system. The system gets constructed iteratively. Let $P_1, \cdots, P_k$ denote the polytopes corresponding to all actions in $s$ except for $a$. We enumerate all combinations $(c_1,\ldots,c_k) \in C(P_1) \times \cdots \times C(P_k)$ of corners of the polytopes. For each such combination we add into $B$ and $d$ new constraints $B_{(c_1,\ldots,c_k)}$ and $d_{(c_1,\ldots,c_k)}$ such that for any $\rho$ satisfying $B_{(c_1,\ldots,c_k)}\rho + d_{(c_1,\ldots,c_k)} \geq 0$ we have $\sum \rho_i c_i \in \paruncertainty{s}{a}{\partitioningord}$. For details, see the procedure $\MinimalProc$ in Algorithm~\ref{fig:algoForBisim}.

\begin{proposition}
We have $B\rho +d \geq 0$ is not feasible if and only if $\paruncertainty{s}{a}{\partitioningord}$ is strictly minimal where the rows of $B$ and $d$ are obtained as a union of rows
\begin{align*}
B \;\; &= \;\; \{\vec{1},\vec{-1}\} \; \cup \; \bigcup \; \{ B_{(c_1,\ldots,c_k)} \; \mid \; (c_1,\ldots,c_k) \;\in\; C(P_1) \times \cdots \times C(P_k) \} \\
d \;\; &= \;\; \{-1,1\} \;\cup\; \bigcup \; \{ d_{(c_1,\ldots,c_k)} \; \mid \; (c_1,\ldots,c_k) \;\in\; C(P_1) \times \cdots \times C(P_k) \}.
\end{align*}
\end{proposition}

\begin{proof}
Let $\rho$ be any feasible solution of the system. It is easy to see that $\paruncertainty{s}{\rho}{\partitioningord} \subseteq \paruncertainty{s}{a}{\partitioningord}$ since $\paruncertainty{s}{\rho}{\partitioningord}$ is convex and since all corners of $\paruncertainty{s}{\rho}{\partitioningord}$ (obtained as a convex $\rho$-combination of corners of all $\paruncertainty{s}{b}{\partitioningord}$) lie within $\paruncertainty{s}{a}{\partitioningord}$. Hence, $\paruncertainty{s}{a}{\partitioningord}$ is not strictly minimal.
As regards the other direction, let $\paruncertainty{s}{a}{\partitioningord}$ be not strictly minimal. By definition, there is a distribution $\rho$ over the remaining actions in $s$ such that $\paruncertainty{s}{\rho}{\partitioningord} \subseteq \paruncertainty{s}{a}{\partitioningord}$. Then, this distribution $\rho$ must satisfy $B\rho + d \geq 0$.
\end{proof}

\paragraph{Complexity}
Again let $\setcardinality{S}=n$, $\setcardinality{A}=m$, $b$ be the maximal number of different actions $\max_{s\in\stateSet} | \{ \intTransitionProbability(s,a,\cdot) \mid a \in\actionSet \} |$, and $f$ be the maximal support of an action $\max_{s\in\stateSet, a\in\actionSet} | \{s' \mid \intTransitionProbability(s,a,s') \neq [0,0] \} |$.

Again, $\ViolateProc_\synthesis$ is called at most $n^3$ times. The procedure $\ViolateProc_\synthesis$ is then linear in $m$ and in the complexity of $\MinimalProc$. There are at most $(f\cdot 2^{f-1})^b$ combinations of corners of the polytopes. For each such combination, $b(b-1)$ times the intersection points of a line and a polytope are computed (in time polynomial in $|\imdp|$), and at most $f!$ facets of the resulting polytope $R$ are inspected. Overall, computing of $\bisim_\synthesis$ can be done in time $|\imdp|^{\bigO{1}} \cdot 2^{\bigO{f^2b}}$.

\section{Case Study}
\label{sec:case-studies}

As a case study, we consider a model of Carrier Sense Multiple Access
with Collision Detection (CSMA / CD), which is an access control on a
shared medium, used mostly in early Ethernet technology.  In this
scheme multiple devices can be connected to a shared bus.  Multiple
attempts at the same time to grab bus access leads to collision.  At
this point, the senders in collision probabilistically schedule a
retransmission according to exponential back-off algorithm.  The
algorithm uniformly determines a delay before the next retransmission,
which is between $0$ to $2^n-1$ time slots after occurrence of $n$-th
collision. After a pre-specified number of failed retransmissions, a
sender aborts the sending process.

There are two sources of uncertainty in the model.  Uncertainty in
sending data lies in the fact that the exact probability of sending a
message from a sender could be unknown.  Instead it is within an
interval.  The other source comes from imprecise information about
collision.  If two nodes try to send a frame at the slightly same
time, a collision will happen.  Conversely it will not happen, when
the later transmitter checks the bus and detects it occupied.  Since
the exact probability of a collision occurrence depends on many
parameters and is likely unknown, it is expressed as an uncertain
interval rather than an exact value in the model.

% Compositional modeling
Concurrent execution of the node and the bus processes assembles the
CSMA/CD model.  To this end we need a formalism that supports
communication among components via parallel composition.  We thus
consider a subclass of abstract
PAs~\cite{DBLP:conf/vmcai/DelahayeKLLPSW11} with interval constraints
on probabilities.  The subclass in general is not closed under
parallel composition.  The problem arises when two actions exhibiting
uncertainty want to synchronise.  Parallel composition in this case
imposes some interdependency between the choices of the composed
action, which cannot be expressed by a simple interval bound and needs
to be expressed by more complicated polynomial constraints.
Nevertheless by excluding synchronisation of actions containing
uncertainty, abstract PAs with interval constraints feature closure
under parallel composition and thereby allow compositional modelling.
This is of course not a strict restriction, because we can always
shift uncertainty to the actions that are not subject to parallel
composition by introducing proper auxiliary states and transitions.
In our case study all components are in this subclass and respects the
restriction, as uncertainty prevails on actions that are not subject
to parallel composition.  Consequently it enables us to utilise
compositional system design by using existing tools.  Since the model
arising from parallel composition is not subject to any further
communication, we can close it and obtain an IMDP at the end.

We use process algebra prCRL~\cite{DBLP:conf/acsd/KatoenPST10},
implemented in tool scoop~\cite{DBLP:conf/qest/Timmer11}, for
compositional modelling of CSMA/CD.  The model has two parameters:
number of nodes attached to the bus and maximum collision allowed
before abortion.  As we are interested in model checking of a model
arisen from parallel composition we apply the semantics of
bisimulation in cooperative way, namely $\bisim_\un$. The state space
is generated by scoop and then the bisimulation quotient is computed.
Since the maximum size $f$ of the set supported by uncertain
transitions is two, the algorithm of
Section~\ref{sec:deciding_strong_bisim} is tractable.  Reduction in
state and transition space gained after bisimulation minimisation is
reported in Table~\ref{tab:reduc}.  As shown in the table, the
reduction of both state and transition space increases when putting
more nodes in the network.  Indeed, then there are more nodes
performing similar activities and thereby increasing the symmetry in
the model.  On the other hand, increasing the maximum number of
collisions allows the nodes to more likely send frames at different
time slots. As a result it decreases the symmetry and then the
reduction factor.
\clearpage
\begin{table}[t]
%  \vspace{-2mm}
  \footnotesize
  \caption{Impact of bisimulation minimisation on CSMA/CD model\protect\footnotemark}
%  \vspace{-1mm}
  \centering
  \begin{tabular}{|c | c | c c | c c | c c |}
    \hline
     \multirow{2}{*}{Node \#} & \multirow{2}{*}{Max collision \#} &  \multicolumn{2}{|c|}{Original Model} &  \multicolumn{2}{|c|}{Minimised Model}  & \multicolumn{2}{|c|}{Reduction Factor}\\
                              &                                   & State \#  & Transition \# & States \# & Transition \#  & For states & For transitions \\ \hline
     \multirow{3}{*}{2}       & 1 & 233     & 466    & 120     & 220    & 48\% & 53\%\\
                              & 2 & 497     & 988    & 310     & 581    & 38\% & 41\%\\
                              & 3 & 865     & 1858   & 576     & 1186   & 33\% & 36\%\\ \hline
     \multirow{3}{*}{3}       & 1 & 4337    & 10074  & 1244    & 2719   & 71\% & 73\%\\
                              & 2 & 52528   & 125715 & 18650   & 42795  & 64\% & 66\%\\
                              & 3 & 239213  & 619152 & 90492   & 225709 & 62\% & 64\%\\ \hline
     \multirow{2}{*}{4}       & 1 & 60357   & 154088 & 10904   & 27308  & 82\% & 82\%\\
                              & 2 & 1832005 & 4876636& 421570  & 1112129& 77\% & 77\%\\
                              \hline
%                              & 3 & & & & & & \\ \hline
     \multirow{1}{*}{5}       & 1 & 751905  & 2043090& 90538   & 248119 & 88\% & 88\% \\
%                              & 2 & & & & & & \\

    \hline
  \end{tabular}\label{tab:reduc}
%\vspace{-4mm}
\end{table}

\footnotetext{ The computation time does not reflect the complexity of the
  algorithm, as it is greatly effected by the file exchange
  between the tools used for modelling and bisimulation minimisation. Hence it is omitted from the table.}

\section{Conclusion}
\label{sec:conclusion}

In this paper, we study strong bisimulations for interval MDPs. In these models there are two sources of non-determinism and we deal with different interpretations of these non-determinisms. This yields two different bisimulations and we give decision algorithms for both of them.

Note that our decision algorithms can be easily adapted to the slightly broader setting of uncertain MDPs with rectangular uncertainty sets~\cite{DBLP:journals/ior/NilimG05}. In this setting, a general convex polytope (not necessarily induced by intervals) is associated to each action in each state. Still, it is assumed that transition probabilities from different states or under different actions are independent.

First open question for future work is the exact complexity of our decision problems. 
One way to address this question is to prove NP-hardness of the general problem. Another way is to identify interesting subclasses of interval MDPs for that a polynomial-time algorithm exists.
Second direction for future work is to address a richer formalism for uncertainties (such as polynomial constraints or even parameters appearing in multiple states/actions).
Third, compositional modelling over interval models also deserves a more systematic treatment. Understanding better the ways how large models with interval uncertainties can be composed, may bring further ideas for efficient analysis of these models.

%
% \vh{I think so. Some possible outlines for future work: 
% \begin{itemize}
% \item Either proving the problem is NP-complete or polynomial 
%\item Compositionality of the model
%\item {Different uncertainty assumptions. Here it was convex intervals. What about general form of uncertainty w.r.t. either constraint-wise uncertainty assumption.}
%\item {In our model nondeterminism btw actions is ignored. How about extension of this model to consider such nondeterminism?}
%\item Effect of this minimization on speed of PCTL model checking algorithm  
%\end{itemize}
%}

\paragraph{Acknowledgements}  We would like to thank Holger Hermanns for inspiring discussions.
This work has been supported by the DFG as
part of the SFB/TR~14 ``Automatic Verification and Analysis of
Complex Systems'' (AVACS), by the Czech Science Foundation under the grant agreement
no.~P202/12/G061, by the DFG/NWO Bilateral Research
Programme ROCKS, and by the European Union Seventh Framework
Programme under grant agreement no.\@ 295261 (MEALS) and 318490 (SENSATION). 
%

%\nocite{*}
\bibliographystyle{eptcs}
%\bibliography{generic}
\bibliography{biblio,bib}

\begin{thebibliography}{10}
\providecommand{\bibitemdeclare}[2]{}
\providecommand{\surnamestart}{}
\providecommand{\surnameend}{}
\providecommand{\urlprefix}{Available at }
\providecommand{\url}[1]{\texttt{#1}}
\providecommand{\href}[2]{\texttt{#2}}
\providecommand{\urlalt}[2]{\href{#1}{#2}}
\providecommand{\doi}[1]{doi:\urlalt{http://dx.doi.org/#1}{#1}}
\providecommand{\bibinfo}[2]{#2}

\bibitemdeclare{inproceedings}{DBLP:conf/concur/AlurHKV98}
\bibitem{DBLP:conf/concur/AlurHKV98}
\bibinfo{author}{Rajeev \surnamestart Alur\surnameend},
  \bibinfo{author}{Thomas~A. \surnamestart Henzinger\surnameend},
  \bibinfo{author}{Orna \surnamestart Kupferman\surnameend} \&
  \bibinfo{author}{Moshe~Y. \surnamestart Vardi\surnameend}
  (\bibinfo{year}{1998}): \emph{\bibinfo{title}{Alternating Refinement
  Relations}}.
\newblock In: {\sl \bibinfo{booktitle}{CONCUR}}, {\sl \bibinfo{series}{LNCS}}
  \bibinfo{volume}{1466}, \bibinfo{publisher}{Springer}, pp.
  \bibinfo{pages}{163--178}.
\newblock \urlprefix\url{http://dx.doi.org/10.1007/BFb0055622}.

\bibitemdeclare{inproceedings}{DBLP:conf/tacas/BenediktLW13}
\bibitem{DBLP:conf/tacas/BenediktLW13}
\bibinfo{author}{Michael \surnamestart Benedikt\surnameend},
  \bibinfo{author}{Rastislav \surnamestart Lenhardt\surnameend} \&
  \bibinfo{author}{James \surnamestart Worrell\surnameend}
  (\bibinfo{year}{2013}): \emph{\bibinfo{title}{LTL Model Checking of Interval
  Markov Chains}}.
\newblock In: {\sl \bibinfo{booktitle}{TACAS}}, {\sl \bibinfo{series}{LNCS}}
  \bibinfo{volume}{7795}, \bibinfo{publisher}{Springer}, pp.
  \bibinfo{pages}{32--46}.
\newblock \urlprefix\url{http://dx.doi.org/10.1007/978-3-642-36742-7_3}.

\bibitemdeclare{book}{Billingsley1979}
\bibitem{Billingsley1979}
\bibinfo{author}{Patrick \surnamestart Billingsley\surnameend}
  (\bibinfo{year}{1979}): \emph{\bibinfo{title}{Probability and Measure}}.
\newblock \bibinfo{publisher}{John Wiley and Sons}, \bibinfo{address}{New York,
  Toronto, London}.

\bibitemdeclare{inproceedings}{CS02}
\bibitem{CS02}
\bibinfo{author}{Stefano \surnamestart Cattani\surnameend} \&
  \bibinfo{author}{Roberto \surnamestart Segala\surnameend}
  (\bibinfo{year}{2002}): \emph{\bibinfo{title}{Decision Algorithms for
  Probabilistic Bisimulation}}.
\newblock In: {\sl \bibinfo{booktitle}{{CONCUR}}}, {\sl \bibinfo{series}{LNCS}}
  \bibinfo{volume}{2421}, pp. \bibinfo{pages}{371--385}.
\newblock \urlprefix\url{http://dx.doi.org/10.1007/3-540-45694-5_25}.

\bibitemdeclare{article}{Chand:1970:ACP}
\bibitem{Chand:1970:ACP}
\bibinfo{author}{Donald~R. \surnamestart Chand\surnameend} \&
  \bibinfo{author}{Sham~S. \surnamestart Kapur\surnameend}
  (\bibinfo{year}{1970}): \emph{\bibinfo{title}{An Algorithm for Convex
  Polytopes}}.
\newblock {\sl \bibinfo{journal}{J. ACM}}
  \bibinfo{volume}{17}(\bibinfo{number}{1}), pp. \bibinfo{pages}{78--86}.
\newblock \urlprefix\url{http://dx.doi.org/10.1145/321556.321564}.

\bibitemdeclare{inproceedings}{DBLP:conf/csl/ChatterjeeCK12}
\bibitem{DBLP:conf/csl/ChatterjeeCK12}
\bibinfo{author}{Krishnendu \surnamestart Chatterjee\surnameend},
  \bibinfo{author}{Siddhesh \surnamestart Chaubal\surnameend} \&
  \bibinfo{author}{Pritish \surnamestart Kamath\surnameend}
  (\bibinfo{year}{2012}): \emph{\bibinfo{title}{Faster Algorithms for
  Alternating Refinement Relations}}.
\newblock In: {\sl \bibinfo{booktitle}{CSL}}, {\sl
  \bibinfo{series}{LIPIcs}}~\bibinfo{volume}{16}, \bibinfo{publisher}{Schloss
  Dagstuhl - Leibniz-Zentrum fuer Informatik}, pp. \bibinfo{pages}{167--182}.
\newblock \urlprefix\url{http://dx.doi.org/10.4230/LIPIcs.CSL.2012.167}.

\bibitemdeclare{inproceedings}{DBLP:conf/fossacs/ChatterjeeSH08}
\bibitem{DBLP:conf/fossacs/ChatterjeeSH08}
\bibinfo{author}{Krishnendu \surnamestart Chatterjee\surnameend},
  \bibinfo{author}{Koushik \surnamestart Sen\surnameend} \&
  \bibinfo{author}{Thomas~A. \surnamestart Henzinger\surnameend}
  (\bibinfo{year}{2008}): \emph{\bibinfo{title}{Model-Checking omega-Regular
  Properties of Interval Markov Chains}}.
\newblock In: {\sl \bibinfo{booktitle}{FoSSaCS}}, {\sl \bibinfo{series}{LNCS}}
  \bibinfo{volume}{4962}, \bibinfo{publisher}{Springer}, pp.
  \bibinfo{pages}{302--317}.
\newblock \urlprefix\url{http://dx.doi.org/10.1007/978-3-540-78499-9_22}.

\bibitemdeclare{article}{DBLP:journals/ipl/ChenHK13}
\bibitem{DBLP:journals/ipl/ChenHK13}
\bibinfo{author}{Taolue \surnamestart Chen\surnameend},
  \bibinfo{author}{Tingting \surnamestart Han\surnameend} \&
  \bibinfo{author}{Marta~Z. \surnamestart Kwiatkowska\surnameend}
  (\bibinfo{year}{2013}): \emph{\bibinfo{title}{On the complexity of model
  checking interval-valued discrete time Markov chains}}.
\newblock {\sl \bibinfo{journal}{Inf. Process. Lett.}}
  \bibinfo{volume}{113}(\bibinfo{number}{7}), pp. \bibinfo{pages}{210--216}.
\newblock \urlprefix\url{http://dx.doi.org/10.1016/j.ipl.2013.01.004}.

\bibitemdeclare{inproceedings}{DBLP:conf/vmcai/DelahayeKLLPSW11}
\bibitem{DBLP:conf/vmcai/DelahayeKLLPSW11}
\bibinfo{author}{Beno\^{\i}t \surnamestart Delahaye\surnameend},
  \bibinfo{author}{Joost-Pieter \surnamestart Katoen\surnameend},
  \bibinfo{author}{Kim~G. \surnamestart Larsen\surnameend},
  \bibinfo{author}{Axel \surnamestart Legay\surnameend},
  \bibinfo{author}{Mikkel~L. \surnamestart Pedersen\surnameend},
  \bibinfo{author}{Falak \surnamestart Sher\surnameend} \&
  \bibinfo{author}{Andrzej \surnamestart Wasowski\surnameend}
  (\bibinfo{year}{2011}): \emph{\bibinfo{title}{Abstract Probabilistic
  Automata}}.
\newblock In: {\sl \bibinfo{booktitle}{VMCAI}}, {\sl \bibinfo{series}{LNCS}}
  \bibinfo{volume}{6538}, \bibinfo{publisher}{Springer}, pp.
  \bibinfo{pages}{324--339}.
\newblock \urlprefix\url{http://dx.doi.org/10.1007/978-3-642-18275-4_23}.

\bibitemdeclare{inproceedings}{DBLP:conf/acsd/DelahayeKLLPSW11}
\bibitem{DBLP:conf/acsd/DelahayeKLLPSW11}
\bibinfo{author}{Beno\^{\i}t \surnamestart Delahaye\surnameend},
  \bibinfo{author}{Joost-Pieter \surnamestart Katoen\surnameend},
  \bibinfo{author}{Kim~G. \surnamestart Larsen\surnameend},
  \bibinfo{author}{Axel \surnamestart Legay\surnameend},
  \bibinfo{author}{Mikkel~L. \surnamestart Pedersen\surnameend},
  \bibinfo{author}{Falak \surnamestart Sher\surnameend} \&
  \bibinfo{author}{Andrzej \surnamestart Wasowski\surnameend}
  (\bibinfo{year}{2011}): \emph{\bibinfo{title}{New Results on Abstract
  Probabilistic Automata}}.
\newblock In: {\sl \bibinfo{booktitle}{ACSD}}, \bibinfo{publisher}{IEEE}, pp.
  \bibinfo{pages}{118--127}.
\newblock
  \urlprefix\url{http://doi.ieeecomputersociety.org/10.1109/ACSD.2011.10}.

\bibitemdeclare{inproceedings}{DBLP:conf/lata/DelahayeLLPW11}
\bibitem{DBLP:conf/lata/DelahayeLLPW11}
\bibinfo{author}{Beno\^{\i}t \surnamestart Delahaye\surnameend},
  \bibinfo{author}{Kim~G. \surnamestart Larsen\surnameend},
  \bibinfo{author}{Axel \surnamestart Legay\surnameend},
  \bibinfo{author}{Mikkel~L. \surnamestart Pedersen\surnameend} \&
  \bibinfo{author}{Andrzej \surnamestart Wasowski\surnameend}
  (\bibinfo{year}{2011}): \emph{\bibinfo{title}{Decision Problems for Interval
  Markov Chains}}.
\newblock In: {\sl \bibinfo{booktitle}{LATA}}, {\sl \bibinfo{series}{LNCS}}
  \bibinfo{volume}{6638}, \bibinfo{publisher}{Springer}, pp.
  \bibinfo{pages}{274--285}.
\newblock \urlprefix\url{http://dx.doi.org/10.1007/978-3-642-21254-3_21}.

\bibitemdeclare{inproceedings}{DBLP:conf/spin/FecherLW06}
\bibitem{DBLP:conf/spin/FecherLW06}
\bibinfo{author}{Harald \surnamestart Fecher\surnameend},
  \bibinfo{author}{Martin \surnamestart Leucker\surnameend} \&
  \bibinfo{author}{Verena \surnamestart Wolf\surnameend}
  (\bibinfo{year}{2006}): \emph{\bibinfo{title}{Don't Know in Probabilistic
  Systems}}.
\newblock In: {\sl \bibinfo{booktitle}{SPIN}}, {\sl \bibinfo{series}{LNCS}}
  \bibinfo{volume}{3925}, \bibinfo{publisher}{Springer}, pp.
  \bibinfo{pages}{71--88}.
\newblock \urlprefix\url{http://dx.doi.org/10.1007/11691617_5}.

\bibitemdeclare{article}{DBLP:journals/ai/GivanLD00}
\bibitem{DBLP:journals/ai/GivanLD00}
\bibinfo{author}{Robert \surnamestart Givan\surnameend},
  \bibinfo{author}{Sonia~M. \surnamestart Leach\surnameend} \&
  \bibinfo{author}{Thomas~L. \surnamestart Dean\surnameend}
  (\bibinfo{year}{2000}): \emph{\bibinfo{title}{Bounded-parameter Markov
  decision processes}}.
\newblock {\sl \bibinfo{journal}{Artif. Intell.}}
  \bibinfo{volume}{122}(\bibinfo{number}{1-2}), pp. \bibinfo{pages}{71--109}.
\newblock \urlprefix\url{http://dx.doi.org/10.1016/S0004-3702(00)00047-3}.

\bibitemdeclare{inproceedings}{DBLP:conf/nfm/HahnHZ11}
\bibitem{DBLP:conf/nfm/HahnHZ11}
\bibinfo{author}{Ernst~Moritz \surnamestart Hahn\surnameend},
  \bibinfo{author}{Tingting \surnamestart Han\surnameend} \&
  \bibinfo{author}{Lijun \surnamestart Zhang\surnameend}
  (\bibinfo{year}{2011}): \emph{\bibinfo{title}{Synthesis for PCTL in
  Parametric Markov Decision Processes}}.
\newblock In: {\sl \bibinfo{booktitle}{NASA Formal Methods}}, {\sl
  \bibinfo{series}{LNCS}} \bibinfo{volume}{6617},
  \bibinfo{publisher}{Springer}, pp. \bibinfo{pages}{146--161}.
\newblock \urlprefix\url{http://dx.doi.org/10.1007/978-3-642-20398-5_12}.

\bibitemdeclare{article}{DBLP:journals/fac/HanssonJ94}
\bibitem{DBLP:journals/fac/HanssonJ94}
\bibinfo{author}{Hans \surnamestart Hansson\surnameend} \&
  \bibinfo{author}{Bengt \surnamestart Jonsson\surnameend}
  (\bibinfo{year}{1994}): \emph{\bibinfo{title}{A Logic for Reasoning about
  Time and Reliability}}.
\newblock {\sl \bibinfo{journal}{Formal Asp. Comput.}}
  \bibinfo{volume}{6}(\bibinfo{number}{5}), pp. \bibinfo{pages}{512--535}.
\newblock \urlprefix\url{http://dx.doi.org/10.1007/BF01211866}.

\bibitemdeclare{techreport}{HHK14}
\bibitem{HHK14}
\bibinfo{author}{Vahid \surnamestart Hashemi\surnameend},
  \bibinfo{author}{Hassan \surnamestart Hatefi\surnameend} \&
  \bibinfo{author}{Jan \surnamestart Kr\v{c}\'{a}l\surnameend}
  (\bibinfo{year}{2014}): \emph{\bibinfo{title}{Probabilistic Bisimulations for
  PCTL Model Checking of Interval MDPs}}.
\newblock \bibinfo{type}{AVACS Technical Report No.} \bibinfo{number}{97},
  \bibinfo{institution}{SFB/TR 14 AVACS}.
\newblock \bibinfo{note}{ISSN: 1860-9821, http://www.avacs.org.}

\bibitemdeclare{article}{DBLP:journals/mor/Iyengar05}
\bibitem{DBLP:journals/mor/Iyengar05}
\bibinfo{author}{Garud~N. \surnamestart Iyengar\surnameend}
  (\bibinfo{year}{2005}): \emph{\bibinfo{title}{Robust Dynamic Programming}}.
\newblock {\sl \bibinfo{journal}{Math. Oper. Res.}}
  \bibinfo{volume}{30}(\bibinfo{number}{2}), pp. \bibinfo{pages}{257--280}.
\newblock \urlprefix\url{http://dx.doi.org/10.1287/moor.1040.0129}.

\bibitemdeclare{inproceedings}{DBLP:conf/lics/JonssonL91}
\bibitem{DBLP:conf/lics/JonssonL91}
\bibinfo{author}{Bengt \surnamestart Jonsson\surnameend} \&
  \bibinfo{author}{Kim~Guldstrand \surnamestart Larsen\surnameend}
  (\bibinfo{year}{1991}): \emph{\bibinfo{title}{Specification and Refinement of
  Probabilistic Processes}}.
\newblock In: {\sl \bibinfo{booktitle}{LICS}}, \bibinfo{publisher}{IEEE
  Computer Society}, pp. \bibinfo{pages}{266--277}.
\newblock \urlprefix\url{http://dx.doi.org/10.1109/LICS.1991.151651}.

\bibitemdeclare{article}{KS90}
\bibitem{KS90}
\bibinfo{author}{Paris~C. \surnamestart Kanellakis\surnameend} \&
  \bibinfo{author}{Scott~A. \surnamestart Smolka\surnameend}
  (\bibinfo{year}{1990}): \emph{\bibinfo{title}{{CCS} {E}xpressions, {F}inite
  {S}tate {P}rocesses, and {T}hree {P}roblems of {E}quivalence}}.
\newblock {\sl \bibinfo{journal}{Inf. Comput.}}
  \bibinfo{volume}{86}(\bibinfo{number}{1}), pp. \bibinfo{pages}{43--68}.
\newblock \urlprefix\url{http://dx.doi.org/10.1016/0890-5401(90)90025-D}.

\bibitemdeclare{inproceedings}{DBLP:conf/formats/KatoenKN09}
\bibitem{DBLP:conf/formats/KatoenKN09}
\bibinfo{author}{Joost-Pieter \surnamestart Katoen\surnameend},
  \bibinfo{author}{Daniel \surnamestart Klink\surnameend} \&
  \bibinfo{author}{Martin~R. \surnamestart Neuh{\"a}u{\ss}er\surnameend}
  (\bibinfo{year}{2009}): \emph{\bibinfo{title}{Compositional Abstraction for
  Stochastic Systems}}.
\newblock In: {\sl \bibinfo{booktitle}{FORMATS}}, {\sl \bibinfo{series}{LNCS}}
  \bibinfo{volume}{5813}, \bibinfo{publisher}{Springer}, pp.
  \bibinfo{pages}{195--211}.
\newblock \urlprefix\url{http://dx.doi.org/10.1007/978-3-642-04368-0_16}.

\bibitemdeclare{inproceedings}{DBLP:conf/acsd/KatoenPST10}
\bibitem{DBLP:conf/acsd/KatoenPST10}
\bibinfo{author}{Joost-Pieter \surnamestart Katoen\surnameend},
  \bibinfo{author}{Jaco \surnamestart van~de Pol\surnameend},
  \bibinfo{author}{Mari{\"e}lle \surnamestart Stoelinga\surnameend} \&
  \bibinfo{author}{Mark \surnamestart Timmer\surnameend}
  (\bibinfo{year}{2010}): \emph{\bibinfo{title}{A Linear Process-Algebraic
  Format for Probabilistic Systems with Data}}.
\newblock In: {\sl \bibinfo{booktitle}{ACSD}}, \bibinfo{publisher}{IEEE
  Computer Society}, pp. \bibinfo{pages}{213--222}.
\newblock
  \urlprefix\url{http://doi.ieeecomputersociety.org/10.1109/ACSD.2010.18}.

\bibitemdeclare{article}{DBLP:journals/rc/KozineU02}
\bibitem{DBLP:journals/rc/KozineU02}
\bibinfo{author}{Igor \surnamestart Kozine\surnameend} \&
  \bibinfo{author}{Lev~V. \surnamestart Utkin\surnameend}
  (\bibinfo{year}{2002}): \emph{\bibinfo{title}{Interval-Valued Finite Markov
  Chains}}.
\newblock {\sl \bibinfo{journal}{Reliable Computing}}
  \bibinfo{volume}{8}(\bibinfo{number}{2}), pp. \bibinfo{pages}{97--113}.
\newblock \urlprefix\url{http://dx.doi.org/10.1023/A:1014745904458}.

\bibitemdeclare{article}{DBLP:journals/iandc/LarsenS91}
\bibitem{DBLP:journals/iandc/LarsenS91}
\bibinfo{author}{Kim~Guldstrand \surnamestart Larsen\surnameend} \&
  \bibinfo{author}{Arne \surnamestart Skou\surnameend} (\bibinfo{year}{1991}):
  \emph{\bibinfo{title}{Bisimulation through Probabilistic Testing}}.
\newblock {\sl \bibinfo{journal}{Inf. Comput.}}
  \bibinfo{volume}{94}(\bibinfo{number}{1}), pp. \bibinfo{pages}{1--28}.
\newblock \urlprefix\url{http://dx.doi.org/10.1016/0890-5401(91)90030-6}.

\bibitemdeclare{book}{DBLP:books/daglib/0067019}
\bibitem{DBLP:books/daglib/0067019}
\bibinfo{author}{Robin \surnamestart Milner\surnameend} (\bibinfo{year}{1989}):
  \emph{\bibinfo{title}{Communication and concurrency}}.
\newblock \bibinfo{series}{PHI Series in computer science},
  \bibinfo{publisher}{Prentice Hall}.

\bibitemdeclare{article}{DBLP:journals/ior/NilimG05}
\bibitem{DBLP:journals/ior/NilimG05}
\bibinfo{author}{Arnab \surnamestart Nilim\surnameend} \&
  \bibinfo{author}{Laurent~El \surnamestart Ghaoui\surnameend}
  (\bibinfo{year}{2005}): \emph{\bibinfo{title}{Robust Control of Markov
  Decision Processes with Uncertain Transition Matrices}}.
\newblock {\sl \bibinfo{journal}{Operations Research}}
  \bibinfo{volume}{53}(\bibinfo{number}{5}), pp. \bibinfo{pages}{780--798}.
\newblock \urlprefix\url{http://dx.doi.org/10.1287/opre.1050.0216}.

\bibitemdeclare{article}{PT87}
\bibitem{PT87}
\bibinfo{author}{Robert \surnamestart Paige\surnameend} \&
  \bibinfo{author}{Robert~E. \surnamestart Tarjan\surnameend}
  (\bibinfo{year}{1987}): \emph{\bibinfo{title}{Three Partition Refinement
  Algorithms}}.
\newblock {\sl \bibinfo{journal}{{SIAM} J. on Computing}}
  \bibinfo{volume}{16}(\bibinfo{number}{6}), pp. \bibinfo{pages}{973--989}.
\newblock \urlprefix\url{http://dx.doi.org/10.1137/0216062}.

\bibitemdeclare{inproceedings}{DBLP:conf/focs/Pnueli77}
\bibitem{DBLP:conf/focs/Pnueli77}
\bibinfo{author}{Amir \surnamestart Pnueli\surnameend} (\bibinfo{year}{1977}):
  \emph{\bibinfo{title}{The Temporal Logic of Programs}}.
\newblock In: {\sl \bibinfo{booktitle}{FOCS}}, \bibinfo{publisher}{IEEE
  Computer Society}, pp. \bibinfo{pages}{46--57}.
\newblock
  \urlprefix\url{http://doi.ieeecomputersociety.org/10.1109/SFCS.1977.32}.

\bibitemdeclare{inproceedings}{DBLP:conf/cav/PuggelliLSS13}
\bibitem{DBLP:conf/cav/PuggelliLSS13}
\bibinfo{author}{Alberto \surnamestart Puggelli\surnameend},
  \bibinfo{author}{Wenchao \surnamestart Li\surnameend},
  \bibinfo{author}{Alberto~L. \surnamestart Sangiovanni-Vincentelli\surnameend}
  \& \bibinfo{author}{Sanjit~A. \surnamestart Seshia\surnameend}
  (\bibinfo{year}{2013}): \emph{\bibinfo{title}{Polynomial-Time Verification of
  PCTL Properties of MDPs with Convex Uncertainties}}.
\newblock In: {\sl \bibinfo{booktitle}{CAV}}, {\sl \bibinfo{series}{LNCS}}
  \bibinfo{volume}{8044}, \bibinfo{publisher}{Springer}, pp.
  \bibinfo{pages}{527--542}.
\newblock \urlprefix\url{http://dx.doi.org/10.1007/978-3-642-39799-8_35}.

\bibitemdeclare{book}{Puterman:1994:MDP:528623}
\bibitem{Puterman:1994:MDP:528623}
\bibinfo{author}{Martin~L. \surnamestart Puterman\surnameend}
  (\bibinfo{year}{1994}): \emph{\bibinfo{title}{Markov Decision Processes:
  Discrete Stochastic Dynamic Programming}}, \bibinfo{edition}{1st} edition.
\newblock \bibinfo{publisher}{John Wiley \& Sons, Inc.}, \bibinfo{address}{New
  York, NY, USA}.
\newblock \urlprefix\url{http://dx.doi.org/10.1002/9780470316887}.

\bibitemdeclare{phdthesis}{mringwal:phdthesis:2009}
\bibitem{mringwal:phdthesis:2009}
\bibinfo{author}{Matthias \surnamestart Ringwald\surnameend}
  (\bibinfo{year}{2009}): \emph{\bibinfo{title}{Reducing Uncertainty in
  Wireless Sensor Networks - Network Inspection and Collision-Free Medium
  Access}}.
\newblock Ph.D. thesis, \bibinfo{school}{ETH Zurich}, \bibinfo{address}{Zurich,
  Switzerland}.

\bibitemdeclare{phdthesis}{Segala-thesis}
\bibitem{Segala-thesis}
\bibinfo{author}{Roberto \surnamestart Segala\surnameend}
  (\bibinfo{year}{1995}): \emph{\bibinfo{title}{Modeling and Verification of
  Randomized Distributed Real-Time Systems}}.
\newblock Ph.D. thesis, \bibinfo{school}{Laboratory for Computer Science,
  Massachusetts Institute of Technology}.

\bibitemdeclare{inproceedings}{SL94}
\bibitem{SL94}
\bibinfo{author}{Roberto \surnamestart Segala\surnameend} \&
  \bibinfo{author}{Nancy~A. \surnamestart Lynch\surnameend}
  (\bibinfo{year}{1994}): \emph{\bibinfo{title}{Probabilistic Simulations for
  Probabilistic Processes}}.
\newblock In: {\sl \bibinfo{booktitle}{{CONCUR}}}, {\sl \bibinfo{series}{LNCS}}
  \bibinfo{volume}{836}, pp. \bibinfo{pages}{481--496}.
\newblock \urlprefix\url{http://dx.doi.org/10.1007/BFb0015027}.

\bibitemdeclare{inproceedings}{DBLP:conf/tacas/SenVA06}
\bibitem{DBLP:conf/tacas/SenVA06}
\bibinfo{author}{Koushik \surnamestart Sen\surnameend}, \bibinfo{author}{Mahesh
  \surnamestart Viswanathan\surnameend} \& \bibinfo{author}{Gul \surnamestart
  Agha\surnameend} (\bibinfo{year}{2006}): \emph{\bibinfo{title}{Model-Checking
  Markov Chains in the Presence of Uncertainties}}.
\newblock In: {\sl \bibinfo{booktitle}{TACAS}}, {\sl \bibinfo{series}{LNCS}}
  \bibinfo{volume}{3920}, \bibinfo{publisher}{Springer}, pp.
  \bibinfo{pages}{394--410}.
\newblock \urlprefix\url{http://dx.doi.org/10.1007/11691372_26}.

\bibitemdeclare{article}{Subramani09}
\bibitem{Subramani09}
\bibinfo{author}{K.~\surnamestart Subramani\surnameend} (\bibinfo{year}{2009}):
  \emph{\bibinfo{title}{On the Complexities of Selected Satisfiability and
  Equivalence Queries over Boolean Formulas and Inclusion Queries over Hulls}}.
\newblock {\sl \bibinfo{journal}{JAMDS}} \bibinfo{volume}{2009}.
\newblock \urlprefix\url{http://dx.doi.org/10.1155/2009/845804}.

\bibitemdeclare{inproceedings}{DBLP:conf/qest/Timmer11}
\bibitem{DBLP:conf/qest/Timmer11}
\bibinfo{author}{Mark \surnamestart Timmer\surnameend} (\bibinfo{year}{2011}):
  \emph{\bibinfo{title}{SCOOP: A Tool for SymboliC Optimisations of
  Probabilistic Processes}}.
\newblock In: {\sl \bibinfo{booktitle}{QEST}}, \bibinfo{publisher}{IEEE
  Computer Society}, pp. \bibinfo{pages}{149--150}.
\newblock
  \urlprefix\url{http://doi.ieeecomputersociety.org/10.1109/QEST.2011.27}.

\bibitemdeclare{inproceedings}{DBLP:conf/cdc/WolffTM12}
\bibitem{DBLP:conf/cdc/WolffTM12}
\bibinfo{author}{Eric~M. \surnamestart Wolff\surnameend}, \bibinfo{author}{Ufuk
  \surnamestart Topcu\surnameend} \& \bibinfo{author}{Richard~M. \surnamestart
  Murray\surnameend} (\bibinfo{year}{2012}): \emph{\bibinfo{title}{Robust
  control of uncertain Markov Decision Processes with temporal logic
  specifications}}.
\newblock In: {\sl \bibinfo{booktitle}{CDC}}, \bibinfo{publisher}{IEEE}, pp.
  \bibinfo{pages}{3372--3379}.
\newblock \urlprefix\url{http://dx.doi.org/10.1109/CDC.2012.6426174}.

\bibitemdeclare{article}{DBLP:journals/ai/WuK08}
\bibitem{DBLP:journals/ai/WuK08}
\bibinfo{author}{Di~\surnamestart Wu\surnameend} \& \bibinfo{author}{Xenofon~D.
  \surnamestart Koutsoukos\surnameend} (\bibinfo{year}{2008}):
  \emph{\bibinfo{title}{Reachability analysis of uncertain systems using
  bounded-parameter Markov decision processes}}.
\newblock {\sl \bibinfo{journal}{Artif. Intell.}}
  \bibinfo{volume}{172}(\bibinfo{number}{8-9}), pp. \bibinfo{pages}{945--954}.
\newblock \urlprefix\url{http://dx.doi.org/10.1016/j.artint.2007.12.002}.

\bibitemdeclare{inproceedings}{yi1994reasoning}
\bibitem{yi1994reasoning}
\bibinfo{author}{Wang \surnamestart Yi\surnameend} (\bibinfo{year}{1994}):
  \emph{\bibinfo{title}{Algebraic Reasoning for Real-Time Probabilistic
  Processes with Uncertain Information}}.
\newblock In: {\sl \bibinfo{booktitle}{FTRTFT}}, {\sl \bibinfo{series}{LNCS}}
  \bibinfo{volume}{863}, \bibinfo{publisher}{Springer}, pp.
  \bibinfo{pages}{680--693}.
\newblock \urlprefix\url{http://dx.doi.org/10.1007/3-540-58468-4_190}.

\end{thebibliography}
\newpage
\end{document}